\documentclass[preprint,12pt]{elsarticle}
\usepackage{etoolbox}
\usepackage{tabularray}
\makeatletter
\def\ps@pprintTitle{%
	\let\@oddhead\@empty
	\let\@evenhead\@empty
	\def\@oddfoot{\reset@font\hfil\thepage\hfil}
	\let\@evenfoot\@oddfoot
}
\makeatother

\usepackage{diagbox}
\usepackage{bm}
\usepackage{amsmath, amsthm, amssymb}
\usepackage{amsfonts}
\usepackage[font=small,labelfont=bf]{caption}
\usepackage{algorithm}
\usepackage{setspace}
\usepackage{comment}

\usepackage[colorlinks=true]{hyperref}
\usepackage{lscape}
\usepackage[utf8]{inputenc} 
\usepackage{amsthm}

\usepackage{enumerate}
\usepackage{mathrsfs}
\usepackage{placeins}
\usepackage{amsfonts}
\usepackage{verbatim}
\usepackage{amssymb}
\usepackage{amsthm}
\usepackage{multirow}
\usepackage{multicol}
\newtheorem{thm}{Theorem}[section]

\theoremstyle{definition}

\usepackage{algpseudocode}

\theoremstyle{remark}
\newtheorem{rem}{Remark}[section]

\numberwithin{equation}{section}

\usepackage[mathscr]{euscript}
\usepackage{geometry} 
\usepackage{graphicx,lscape} 
\usepackage{float}
\usepackage{booktabs} 
\usepackage{bigstrut}
\usepackage{array} 
\usepackage{paralist} 
\usepackage{verbatim} 
\usepackage{subfig} 

\biboptions{comma,round,authoryear}

\usepackage{relsize}%
\usepackage{listings}
\usepackage[normalem]{ulem}
\useunder{\uline}{\ul}{}

\usepackage{natbib}

\usepackage{hyperref}
\hypersetup{
colorlinks=true,
allcolors=blue,
pdftitle={A novel INAR(1) process based on the new McKenzie - error thinning operator with geometric marginals},
pdfpagemode=FullScreen,
}

\usepackage{stackengine}

\usepackage[titletoc]{appendix}

\usepackage{enumitem}
\newlist{steps}{enumerate}{1}
\setlist[steps, 1]{label = Step \arabic*:}

\newlist{notes}{enumerate}{1}
\setlist[notes]{label=Note: ,leftmargin=*}

\begin{document}
\begin{frontmatter}
	
	
	
	\title{\textbf{Bayesian estimation for novel geometric INGARCH model}}
			\author[label1]{Divya Kuttenchalil Andrews\corref{cor1}}
		\ead{divyaandrews5@gmail.com}
		\author[label2]{N. Balakrishna}
		\address[label1]{Cochin University of Science and Technology, Kochi, India.}
		\address[label2]{Indian Institute of Technology, Tirupati, India.}
	\author{}
	\address{}
	\begin{abstract}
		This paper introduces an integer-valued generalized autoregressive conditional heteroskedasticity (INGARCH) model based on the novel geometric distribution and discusses some of its properties. The parameter estimation problem of the models are studied by conditional maximum likelihood and Bayesian approach using Hamiltonian Monte Carlo (HMC) algorithm. The results of the simulation studies and real data analysis affirm the good performance of the estimators and the model. Forecasting using the Bayesian predictive distribution has also been studied and evaluated using real data analysis.\\
		
		\noindent Keywords:	Bayesian inference; Hamiltonian Monte Carlo; INGARCH; count time series.\\\\

	\end{abstract}
\end{frontmatter}

\section{Introduction}

\noindent 
Count time series are often encountered in practical applications, particularly in fields such as insurance, healthcare, epidemiology, queuing models, communications, reliability studies, and meteorology. Examples of such counts include the number of patients, crime incidents, transmitted messages etc. Several models have been designed to analyze count time series, focusing on their marginal distributions and autocorrelation structures. These models are typically divided into two broad categories: models based on the 'thinning' operator and those of the regression type. In this paper, we consider the latter class of observation-driven models: the integer-valued generalized autoregressive conditional heteroscedastic (INGARCH) models introduced by \cite{heinen2003modelling} and \cite{ferland2006integer} and defined as follows:
\begin{equation}
	\label{pingarch}
	\left\{\begin{array}{l}
		X_t \mid \mathscr{F}_{t-1}:  Poisson(\lambda_t),  \\
		\lambda_t=\alpha_0+\sum_{i=1}^p\alpha_i X_{t-i}+\sum_{j=1}^{q}\beta_j \lambda_{t-j},
	\end{array}\right.
\end{equation}
where	$\mathscr{F}_{t-1}$ is the $\sigma$ - field generated by $\{X_{t-1}, X_{t-2},\ldots\}$,  $\alpha_0 > 0$, $\alpha_i \geq 0$, $\beta_j \geq 0$, for $i=1,\ldots p$, $j = 1, \ldots q$, $p \geq 1$, $q \geq 1$.  If $q=0$, (\ref{pingarch}) is referred to as INARCH(p) (or INGARCH(p,0)) model. The term INGARCH was introduced by \cite{ferland2006integer} due to its similarity to the classical GARCH model proposed by \cite{bollerslev1986generalized}:
\begin{equation} 
	\left\{\begin{array}{l}
		X_t =  \sigma_t.\varepsilon_t,  \\
		\sigma_t^2=\alpha_0+\sum_{i=1}^p\alpha_i X_{t-i}^2+\sum_{j=1}^{q}\beta_j \sigma_{t-j}^2,
	\end{array}\right. \nonumber
\end{equation}
where $\sigma_t^2 = Var[X_t \mid \mathscr{F}_{t-1}]$ and $\{\varepsilon_t\}$ is a sequence of white noise with mean 0 and variance 1. Further, $\{X_s\}$ is independent of $\{\varepsilon_t\}$ for every $s<t$. The model (\ref{pingarch}) is referred to as the Poisson INGARCH (PINGARCH) model  where $\lambda_t = E[X_t\mid \mathscr{F}_{t-1}]$. Eventhough $X_t$ conditioned on $\mathscr{F}_{t-1}$ follows the equidispersed Poisson distribution with mean $\lambda_t$, it can model overdispersed counts (See \cite{weissovrdisp}):
\begin{equation}\nonumber
	\mu_t = E[X_t] = E[\lambda_t], \quad Var[X_t] = \mu_t + Var[\lambda_t] > \mu_t.
\end{equation}
\cite{ferland2006integer} showed that if $\sum_{i=1}^p \alpha_i + \sum_{j=1}^q \beta_j <1$ is satisfied, the INGARCH process exists and is strictly stationary, with finite first and second - order moments. Then, the unconditional mean of $X_t$ is given by
\begin{equation}
	\mu = \frac{\alpha_0}{1- \sum_{i=1}^p\alpha_i - \sum_{j=1}^q \beta_j}. 
\end{equation}
Further properties on variance and autocovariances as well as an in-depth review of the purely autoregressive case i.e PINARCH(p) (or PINGARCH(p,0) ) was conducted by \cite{weiss2009modelling}. 
\cite{zhu2011negative} proposed the negative binomial INGARCH (NB-INGARCH) model of the form:
\begin{equation}
	\label{nbgarch}
	\left\{\begin{array}{l}
		X_t \mid \mathscr{F}_{t-1}:  NB (n,p_t),  \\
		\frac{1-p_t}{p_t}=\lambda_t=\alpha_0+\sum_{i=1}^p\alpha_i X_{t-i}+\sum_{j=1}^{q}\beta_j \lambda_{t-j},
	\end{array}\right.
\end{equation}
where	$\mathscr{F}_{t-1}$ , $\alpha_0 ,\alpha_i ,\beta_j $, reprise the definition and conditions in (\ref{pingarch})  respectively. The parameter $n$ is considered to be fixed, while $p_t$ varies with time and $p_t = 1/(1+\lambda_t)$. The conditional mean and variance respectively are
\smaller
\begin{equation}
	E[X_t\mid \mathscr{F}_{t-1}] = \frac{n(1-p_t)}{p_t} = n\lambda_t, 
	\; \textrm{and} \;
	Var[X_t\mid\mathscr{F}_{t-1}] = \frac{n(1-p_t)}{p_t^2} = n \lambda_t(1+\lambda_t). 
\end{equation}
\normalsize
An alternate version of the negative binomial INGARCH was proposed by \cite{xu2012model} with the conditional distribution $NB (n_t,p)$ with $n_t = \lambda_t\frac{p}{1-p}$ and with the conditional mean $\lambda_t$ satisfying right hand side of (\ref{nbgarch}). In this article, we consider the model by  \cite{zhu2011negative} for data analysis and comparison. \cite{zhu2012modeling} defined the generalized Poisson (GP - INGARCH) model as
\begin{equation}
	\label{gpingarch}
	\left\{\begin{array}{l}
		X_t \mid \mathscr{F}_{t-1}:  GP(\eta_t,\kappa),  \\
		\frac{\eta_t}{1-\kappa}=\lambda_t=\alpha_0+\sum_{i=1}^p\alpha_i X_{t-i}+\sum_{j=1}^{q}\beta_j \lambda_{t-j},
	\end{array}\right.
\end{equation}
where $p \geq 1$, $q \geq 1$, $\alpha_0 >0$, $max(-1,-\eta _t/4) < \kappa < 1$, $\alpha_i , \beta_j \geq 0, \; i=1,\ldots,p, \; j= 1,\ldots,q$, and $\mathscr{F}_{t-1}$ is the past information. 
The probability mass function (pmf) of the generalized Poisson distribution ($GP(\eta,\kappa)$) is defined as
\begin{equation}
	Pr[X=x]=\begin{cases}
		\eta (\eta + \kappa x)^{x-1}e^{-(\eta + \kappa x)}/x!,	 \;x = 0,1,\ldots ,\\
		0,\quad \quad  \quad \quad \quad \quad \quad \quad \quad \textrm{for} \;x > m \; \textrm{if} \;  \kappa <0, \nonumber
	\end{cases} 
\end{equation}
where $\eta >0$, $max(-1,-\eta/m) < \kappa <1$, and  $m (\geq 4)$ is the largest positive integer for which $\eta + \kappa m >0 $ when $\kappa < 0$. It reduces to the usual Poisson with parameter $\eta$ when $\kappa = 0$. Note that when $\kappa < 0$, the distribution is truncated because $Pr[X = x] = 0$ for all $x > m$, making the sum $\sum_{x=0}^{m} Pr[X = x]$ slightly less than 1. However, this truncation error is less than $0.5\%$ when $m \geq 4$, so it is negligible in practical applications \citep{consul2006lagrangian}. The respective conditional mean and conditional variance of $\{X_t\}$  following (\ref{gpingarch}) are
\begin{equation}
	E[X_t\mid\mathscr{F}_{t-1}]= \frac{\eta_t}{1-\kappa} = \lambda_t, 
	\; \textrm{and} \;
	Var[X_t\mid\mathscr{F}_{t-1}] = \frac{\eta_t}{(1-\kappa)^3} = \frac{\lambda_t}{(1-\kappa)^2} .
\end{equation}
Detailed discussions and review of INGARCH models can be found in \cite{weiss2018introduction} and \cite{liu2023systematic}. \par
In the following section, we introduce the novel geometric INGARCH (NoGe-INGARCH) model and discuss some of its properties. The proposed NoGe-INGARCH model is novel in its synthesis of a flexible, zero-inflated count distribution with an autoregressive conditional mean structure tailored for time series. While the literature on count time series includes Poisson, negative binomial, and generalized Poisson INGARCH variants, this is, to our knowledge, the first work to embed a structurally zero-inflated geometric model within an INGARCH framework. In addition to modeling innovation, our work contributes a computational advance by implementing full Bayesian inference via \textit{Hamiltonian Monte Carlo (HMC)} using the \textit{Stan} probabilistic programming language, enabling efficient posterior sampling even for complex, constrained parameter spaces. We emphasize graphical diagnostics such as posterior trace plots, autocorrelation functions of residuals, and probability integral transform (PIT) histograms---tools that are essential for interpreting model adequacy and convergence in a Bayesian setting.

\section{Novel geometric INGARCH model}\label{nogeingarch}
The novel geometric distribution is defined by the pmf: 
\begin{equation}
	\label{eq1}
	Pr[X=x] =\delta \phi+ (1-\delta)(1-\phi) \left(1-\theta\right)^{x-1}\theta, \quad x \in \mathbb{N}_0, 
\end{equation}
where  $0<\phi<1$ and $0<\theta \leq 1$. The indicator function $\delta$  assumes the value $1$ when $x=0$ and $0$ otherwise. The term ``novel geometric" is carried forward from \cite{doi:10.1080/00949655.2023.2213794}. We propose the NoGe-INGARCH  model :
\begin{equation}
	\label{eq2}
	\left\{\begin{array}{l}
		X_t \mid \mathscr{F}_{t-1}:  NoGe\left(\theta_t, \phi \right),  \\
		\frac{1-\phi}{\theta_t} = \lambda_t=\alpha_0+\sum_{i=1}^p\alpha_i X_{t-i}+\sum_{j=1}^{q}\beta_j \lambda_{t-j},
	\end{array}\right.
\end{equation}
where $\alpha_0 > 0$, $\alpha_i \geq 0$, $\beta_j \geq 0$, for $i=1,\ldots p$, $j = 1, \ldots q$, $p \geq 1$, $q \geq 1$. 

The motivation for employing the novel geometric (NoGe) distribution in the INGARCH framework stems from its capacity to handle both \textit{overdispersion} and \textit{zero inflation}, which are frequently observed in real-world count data. Unlike the Poisson distribution, which assumes equidispersion, or the standard geometric distribution, which lacks explicit control over structural zeros, the NoGe distribution introduces a point mass at zero through the parameter $\phi$. This added flexibility makes it particularly suitable for datasets such as epidemiological counts or number of insurance claims, where excessive zeros occur not merely due to chance but as an inherent structural feature. The distribution retains mathematical tractability while providing a more nuanced model of dispersion and data sparsity. Moreover, the conditional moments derived from the NoGe-INGARCH model maintain closed forms, thereby facilitating both theoretical analysis and computational implementation.

The conditional mean and conditional variance of a NoGe-INGARCH(p,q) process $\{X_t\}$ are:
\smaller
\begin{equation}
	\label{eqr}
	E[X_t\mid\mathscr{F}_{t-1}] = \frac{1-\phi}{\theta_t} = \lambda_t ,\; \textrm{and} \; Var[X_t\mid \mathscr{F}_{t-1}] = \frac{1-\phi}{\theta_t}\left(\frac{1+\phi}{\theta_t} - 1 \right) = \lambda_t\left(\frac{1+\phi}{1-\phi}\lambda_t -1\right). 
\end{equation}
\normalsize
\begin{thm}
	\label{thm1}
	A necessary and sufficient condition for the NoGe-INGARCH(p,q) model  (assuming $p>q$) to be stationary in the mean is that the roots of the equation 
	\begin{equation}
		\label{eq7}
		1- \sum_{i=1}^{q}\left(\alpha_i + \beta_i \right) b^{-i} - \sum_{i=q+1}^p  \alpha_i b^{-i} = 0,
	\end{equation}
	all lie inside the unit circle. 
\end{thm}
\begin{proof}
We begin by taking expectations of both sides of the NoGe-INGARCH(p,q) model defined by \eqref{eq2}, assuming $\mu_t = E[X_t]$.
	\begin{equation}
		\begin{aligned}\label{pf11}
			\mu_t &= E[X_t] = E[\lambda_t] \\
			&= \alpha_0 + \sum_{i=1}^p \alpha_iE[X_{t-i}] + \sum_{i=q+1}^p \beta_j E[\lambda_{t-j}]\\
			&= \alpha_0 + \sum_{i=1}^p \alpha_i\mu_{t-i}+\sum_{i=q+1}^p\beta_j\mu_{t-j}.
		\end{aligned}
	\end{equation}
Rearranging \eqref{pf11} to the form of a non-homogeneous difference equation  :
		\begin{equation}\label{pf12}
			\mu_t - \sum_{i=1}^p \alpha_i \mu_{t-i} - \sum_{j=1}^q \beta_j \mu_{t-j} = \alpha_0. 
		\end{equation}
If there exists a constant solution $\mu_t = \mu$ for all $t$, then $\mu$ is called an \textit{equilibrium} or \textit{stationary value}. Substituting $\mu_t = \mu$ into \eqref{pf12}, we get 
\newline
\begin{equation}
	\mu - \sum_{i=1}^p \alpha_i \mu - \sum_{j=1}^q \beta_j \mu = \alpha_0,
\end{equation}
\normalsize
which simplifies to:
\begin{equation}
	\mu =\frac{\alpha_0}{1- \sum_{i=1}^p\alpha_i - \sum_{j=1}^q \beta_j}, \quad \text{provided } 1- \sum_{i=1}^p\alpha_i - \sum_{j=1}^q \beta_j \neq 0.
	\label{eq:equilibrium}
\end{equation}
The equilibrium $\mu$ is said to be \textit{stable} if every solution of \eqref{pf12}, regardless of the initial values $\mu_0$ and $\mu_1$, converges to $\mu$ as $t \to \infty$:
\begin{equation}
	\lim_{t \to \infty} \mu_t = \mu.
	\label{eq:stability_condition}
\end{equation}
To analyze stability, we define a new sequence $y_t$ to represent deviations from equilibrium:
\begin{equation}
	y_t = \mu_t - \mu.
	\label{eq:deviation}
\end{equation}
We then have
\begin{align*}
	y_{t} - \sum_{i=1}^p \alpha_i y_{t-i}- \sum_{j=1}^q \beta_j y_{t-j} &= \left(\mu_{t} - \sum_{i=1}^p \alpha_i \mu_{t-i}- \sum_{j=1}^q \beta_j \mu_{t-j} \right) - \left(1- \sum_{i=1}^p\alpha_i - \sum_{j=1}^q \beta_j\right) \mu \\
	&=  \mu_{t} - \sum_{i=1}^p \alpha_i \mu_{t-i}- \sum_{j=1}^q \beta_j \mu_{t-j} - \alpha_0 = 0.
\end{align*}
Thus, the sequence $\{y_t\}$ satisfies a homogeneous difference equation:
\begin{equation}
	y_{t} - \sum_{i=1}^p \alpha_i y_{t-i}- \sum_{j=1}^q \beta_j y_{t-j} = 0.
	\label{eq:homogeneous}
\end{equation}
The auxilliary (characteristic) equation corresponding to \eqref{eq:homogeneous}, considering 
$p>q$ is given by:
\smaller
\begin{equation}\label{pf13}
	1- \sum_{i=1}^{q}\left(\alpha_i + \beta_i \right) b^{-i} - \sum_{i=q+1}^p  \alpha_i b^{-i} = 0.
\end{equation}
\normalsize
For $\mu_t$ to converge to a finite limit $\mu$ (stationarity), the sequence $y_t$, which represents the deviation of $\mu_t$ from its equilibrium value $\mu$, converges to zero for any given pair of initial values $y_0$ and $y_1$. Since $y_t$ satisfies the homogeneous difference equation \eqref{eq:homogeneous}, the proof now follows from the necessary and sufficient condition, discussed in \cite{goldberg58}, for a  homogeneous difference equation to have a stable solution converging to zero which is independent of initial values, viz., the roots $b_1,\ldots,b_p$ of (\ref{pf13}) all lie inside the unit circle.
\end{proof}
The following theorem states the second - order stationarity conditions for the NoGe-INGARCH model. To elucidate the core concept of the proof, a simplified NoGe-INARCH(p) model is considered.
\begin{thm}
	\label{thm2}
	Suppose that the process $\{X_t\}$ following NoGe-INARCH(p) model is first - order stationary. Then, a necessary and sufficient condition for the process to be second - order stationary is that 
	\begin{equation}
		1- L_1b^{-1} - \ldots - L_p b^{-p} = 0 \nonumber
	\end{equation}
	has all roots lying inside the unit circle, where for $r,s= 1, \ldots, p-1$, 
	$L_r = \frac{2}{1-\phi}\left(\alpha_r^2 - \sum_{v=1}^{p-1}\sum_{|i-j|=v} \alpha_i\alpha_jm_{vr}\nu_{r0}\right)$, $L_p =  \frac{2\alpha_p^2}{1-\phi}$, $\nu_{s0} = \alpha_s$, $\nu_{ss} = \sum_{|i-s|=s} \alpha_i -1$ and $\nu_{sr} = \sum_{|i-s|=r}\alpha_i, r \neq s$. Also, $M$ and $M^{-1}$ are $(p-1) \times (p-1)$ matrices such that $M=(\nu_{ij})_{i,j=1}^{p-1}$ and  $M^{-1}=(m_{ij})_{i,j=1}^{p-1}$.
\end{thm} 

The proof is given in \textcolor{blue}{Appendix} \ref{pf1}. \par
 Unlike existing INGARCH formulations where second-order conditions are derived under simpler marginal assumptions, the NoGe-INGARCH model necessitates more careful treatment due to the interaction of dispersion and autoregression parameters.
\par The next theorem states the structure of autocovariances of the NoGe-INGARCH model, given that it is second - order stationary.
\begin{thm}
	\label{thm3}
	Suppose that $\{X_t\}$ following NoGe-INGARCH(p,q) process is second - order stationary. Let the autocovariances be defined as:
	\begin{eqnarray}
		\gamma_X(h) &=& Cov[X_t, X_{t-h}], \text{and}\nonumber\\
		\gamma_{\lambda}(h) &=& Cov[\lambda_t, \lambda_{t-h}]. \nonumber
	\end{eqnarray} Then, they satisfy the equations
	\smaller
	\begin{equation}
		\begin{aligned}
			\gamma_X(h) &= \sum_{i=1}^p\alpha_i\gamma_X(|h-i|) + \sum_{j=1}^{min(h-1,q)}\beta_j\gamma_X(h-j) +  \sum_{j=h}^q\beta_j\gamma_{\lambda}(j-h), \quad h \geq 1; \nonumber \\
			\gamma_{\lambda}(h) &=  \sum_{i=1}^{min(h,p)}\alpha_i\gamma_{\lambda}(|h-i|) +  \sum_{i=h+1}^p\alpha_i\gamma_X(i-h) + \sum_{j=1}^q \beta_j \gamma_{\lambda}(|h-j|), \quad h \geq 0.
		\end{aligned}
	\end{equation}
\end{thm}
The proof is detailed in \textcolor{blue}{Appendix} \ref{pf}. 
\begin{rem}
Consider the NoGe-INGARCH (1,1) model:
\begin{equation}
	\label{eq14}
	\left\{\begin{array}{l}
		X_t \mid \mathscr{F}_{t-1}:  NoGe\left(\theta_t, \phi\right),  \\
		\frac{1-\phi}{\theta_t} = \lambda_t=\alpha_0+\alpha_1 X_{t-1}+\beta_1 \lambda_{t-1}, \nonumber
	\end{array}\right.
\end{equation}
where $\alpha_0 > 0$, $\alpha_1\geq 0$ and $\beta_1\geq 0$.
By \textcolor{blue}{Theorem} \ref{thm1}, if the process $\{X_t\}$ is stationary, we have 
\begin{equation}
	\mu = E[X_t] = \frac{\alpha_0}{1-\alpha_1 - \beta_1}. \nonumber
\end{equation}
From \textcolor{blue}{Theorem} \ref{thm3}, we obtain
\begin{equation}
	\gamma_X(h) = (\alpha_1 + \beta_1)^{h-1} \gamma_X(1), \; h\geq 2, \;\textrm{and}\; \gamma_{\lambda}(h)=(\alpha_1 + \beta_1 )^h \gamma_{\lambda}(0), \; h \geq 1, \nonumber
\end{equation}
where $\gamma_X(1) = \big(\frac{2 \alpha_1}{1-\phi} + \beta_1 \big)\gamma_{\lambda}(0) + \alpha_1\mu\big(\frac{1+\phi}{1-\phi}\mu - 1\big)$,  and $\gamma_{\lambda}(0) = \frac{\alpha_1^2\mu\big((1+\phi)\mu - (1-\phi)\big)}{(1-\phi)\big(1 - (\zeta \alpha_1 ^2 +2 \alpha_1\beta_1+ \beta_1^2)\big)}.$ Then, for $\zeta = \frac{2}{1-\phi}$, the unconditional variance, $Var[X_t] = \left(\frac{(1+\phi)}{(1-\phi)}\mu^2 - \mu \right)\frac{1- (\zeta-1)\alpha_1^2 -2\alpha_1\beta_1 - \beta_1^2}{1- \zeta \alpha_1^2 - 2\alpha_1\beta_1 - \beta_1^2}$. Further, the autocorrelations  $\rho_{X}(h)= (\alpha_1 +\beta_1)^h \frac{\alpha_1(1-\alpha_1\beta_1-\beta_1^2)}{1-(\zeta -1) \alpha_1^2 - 2\alpha_1\beta_1-\beta_1^2}, \; h \geq 1$, and $\rho_{\lambda}(h) = (\alpha_1 +\beta_1)^h, \; h\geq 0$. The next section delves into the estimation techniques used for estimation of parameters of the model.
\end{rem}
\section{Estimation}\label{est}
In the present study, we have used two methods for estimation of parameters --- conditional maximum likelihood and Bayesian estimation. \\
Let $\boldsymbol{\alpha} =  (\alpha_0,\alpha_1, \ldots \alpha_p)^T$, $\boldsymbol{\beta} = (\beta_1,\ldots, \beta_q)^T$, $\boldsymbol{\theta}= (\boldsymbol{\alpha}^T,\boldsymbol{\beta}^T)^T=(\theta_0,\theta_1,\ldots \theta_{p+q})^T$, $\boldsymbol{\Theta} =  (\phi,\boldsymbol{\theta})^T= (\Theta_1,\Theta_2,\ldots \Theta_{p+q+2})^T$.
To estimate $\boldsymbol{\Theta}$, we will first define the conditional log - likelihood function 
\begin{equation}
	L(\boldsymbol{\Theta}\mid  \mathscr{F}_{t-1}) =\prod_{t=1}^{n} Pr[	X_t = x_t \mid \mathscr{F}_{t-1}]= \prod_{t=1}^{n} \phi ^{\delta_t} \left[(1-\phi)^2\left(1-\frac{1-\phi}{\lambda_t}\right)^{x_t - 1} \frac{1}{\lambda_t}\right]^{(1-\delta_t)},
\end{equation}
where $\lambda_t$ is updated according to definition in (\ref{eq2}), $\delta_t$ is the indicator function assuming $1$ when $x_t=0$ and $0$ otherwise, and  $0<\phi<1$ . This implies that the conditional log - likelihood is of the form
\smaller
\begin{equation}
	\label{logl1}
	\begin{aligned}
		l(\boldsymbol{\Theta}\mid  \mathscr{F}_{t-1})&=log L(\boldsymbol{\Theta}\mid  \mathscr{F}_{t-1})  \\
		&= \sum_{t=1}^{n} \left\{\delta_t log\phi + (1-\delta_t)\left[ 2log(1-\phi) + (x_t - 1) log\left(1-\frac{1-\phi}{\lambda_t}\right) - log(\lambda_t) \right]\right\}.
	\end{aligned}
\end{equation}		
\normalsize
The conditional maximum likelihood estimates (CMLEs) of the parameter vector $\boldsymbol{\Theta}$  can be found by maximizing $l(\boldsymbol{\Theta}\mid  \mathscr{F}_{t-1})$ specified in (\ref{logl1}). However, it is easy to see that the estimates have no closed form and numerical optimization methods have to be used. In the following subsection, we develop an algorithm for numerical optimization of log-likelihood function using Bayesian methods.
\subsection{The Hamiltonian Monte Carlo (HMC) algorithm}\label{hmc}
Let $\mathscr{P}_0(\boldsymbol{\Theta})$ denote the prior density of $\boldsymbol{\Theta}$. By Bayes' rule, the posterior density of $\boldsymbol{\Theta}$ given the data is proportional to the product of the likelihood function and the prior, i.e.,
\begin{equation}
	\label{post}
	\mathscr{P}(\boldsymbol{\Theta} \mid \mathscr{F}_{t-1}) \propto L(\boldsymbol{\Theta} \mid \mathscr{F}_{t-1} ) \mathscr{P}_0(\boldsymbol{\Theta}).
\end{equation}
In general, we assume multivariate normal non-informative priors  with large variances for the parameters \citep{andrade2023zero}. However, one must bear in mind that for the model proposed in (\ref{eq2}), all parameters are assumed non-negative or strictly positive (in the case of $\alpha_0$) real-valued. In addition, constraints imposed by the stationarity conditions on the parameters and that the parameter $\phi$ is always bounded between 0 and 1 should also be speculated. So, we transform the constrained parameters to the real space by suitable functions such as log-odds (logit) , square root etc. depending on the bounds imposed by the conditions. Then, we proceed to assign normal priors to the transformed unconstrained parameters. \par In the present article, we have considered logit transformations, where $logit(x) = log(\frac{x}{1-x}), \; x \in (0,1)$, to the constrained parameters with the exception of the intercept $\alpha_0$ for which a log-normal prior is assumed. That is,
\[
(\alpha_1^{'},\ldots,\alpha_p^{'}, \beta_1^{'}, \ldots, \beta_q^{'})^{T} \equiv (\boldsymbol{\alpha^{'},\beta^{'}}) ^T \sim N(\boldsymbol{\mu_{\alpha^{'},\beta^{'}}, \Sigma_{\alpha^{'},\beta^{'}}}), 
\]
\[
\alpha_0 \sim lognormal(\mu_{\alpha_0}, \sigma_{\alpha_0}^2), \quad
\phi^{'} \sim N(\mu_\phi, \sigma_\phi^2 ),
\]
where $\boldsymbol{\alpha}^{'} = (\alpha_1^{'},\ldots,\alpha_p^{'})$,  $\boldsymbol{\beta}^{'} = (\beta_1^{'}, \ldots, \beta_q^{'})$, and $\phi^{'}$ denote the logit - transformed parameters. For parameters bounded in the open interval $(l,u)$, we can use the scaled log-odds transformation given by $ \theta ^{'} = \text{logit} \left( \frac{\theta - l}{u - l} \right)$. 
The inverse of this transform is $\theta = l + (u - l) \cdot \text{logit}^{-1}(\theta^{'})$. The task that remains is to draw samples from the posterior distribution. \par
Statistical models with complex posterior distributions often rely on Markov chain Monte Carlo (MCMC) methods, detailed by \cite{gelfand1990sampling}, to sample from these distributions. The Metropolis-Hastings(MH) algorithm, discussed by \cite{tierney1994markov} as a versatile method for creating Markov chains from the target distribution $\mathscr{P}(\boldsymbol{\Theta})$, is commonly employed. The algorithm employs a ``jumping distribution" $\mathfrak{p}(\boldsymbol{\Theta}^* | \boldsymbol{\Theta})$ to enable transitions from a current parameter state $\boldsymbol{\Theta}$ to a proposed state $\boldsymbol{\Theta}^*$ within the parameter space. These proposed transitions are accepted with a probability $\min \left(1, \frac{\mathfrak{p}(\boldsymbol{\Theta}^* | \boldsymbol{\Theta})}{\mathfrak{p}(\boldsymbol{\Theta} | \boldsymbol{\Theta}^*)}\right)$. A likely choice for $\mathfrak{p}(.)$ is a normal distribution rendering the sequence of samples into a Gaussian random walk \citep{chib1995understanding}. \par The random walk Metropolis is easy to implement and has an intuitive appeal. It tends to propose points in regions with large volumes, often pushing proposals towards the tails of the target distribution. The Metropolis correction step then rejects proposals that fall into low-density areas, favoring those within high-probability regions, thereby concentrating towards the typical set of the target distribution. However, as the dimensionality of the target distribution increases, the volume outside the typical set becomes dominant, causing most proposals to land in low-density tail regions, resulting in very low acceptance probabilities and frequent rejections. Reducing the step size, $\epsilon$, (specified based on the required level of accuracy) of proposals can increase acceptance by staying within the typical set, but this slows down the movement of the Markov chain significantly leading to highly biased MCMC estimators. \citep{betancourt2017conceptual}
\par The HMC algorithm, initially known as Hybrid Monte Carlo \citep{duane1987hybrid}, extends the MH algorithm by generating more accurate proposal values through the use of Hamiltonian dynamics. It introduces a momentum variable $\nu_j$ for each component $\Theta_j$ in the target space. These variables $\Theta_j$ and $\nu_j$ are simultaneously updated using a modified MH algorithm, where the jumping distribution for $\Theta$ is heavily influenced by $\nu$. Essentially, the momentum $\nu$ indicates the expected distance and direction of jump in $\Theta$, promoting consecutive jumps in a consistent direction and facilitating rapid movement across the $\Theta$ space when feasible. This is implemented using the leapfrog algorithm. Thus, it succeeds in suppressing the random walk behaviour. The MH accept/reject rule halts movement upon entering low-probability regions, prompting momentum adjustments until transition can resume \citep{gelman2013bayesian}.

In short, in the above sampling strategy, proposal distributions are directed towards the mode(s) of the posterior distribution rather than being symmetrical around the current position. These proposals utilize trajectories derived from the gradient of the posterior, and then employs the MH method to accept or reject these choices \citep{neal2011}. Thus, compared to MH algorithms, HMC offers a more efficient approach to Monte Carlo sampling.\citep{kruschke2014doing}

The Hamiltonian function (See \cite{andrade2023zero}) can be written as :
\begin{equation}
	\label{ham}
	\mathcal{H}(\boldsymbol{\Theta}, \boldsymbol{\nu}) = \mathcal{U}(\boldsymbol{\Theta}) + \mathcal{K}(\boldsymbol{\nu}),
\end{equation}
where $\mathcal{H}(.)$ denotes the total energy of the system, $\mathcal{U(.)}$ represents the potential energy  and $\mathcal{K}$ the kinetic energy respectively. The changes in $\boldsymbol{\Theta}$ and $\boldsymbol{\nu}$ over iterations $\tau$ are governed by Hamilton's equations, which are derived from the partial derivatives of the Hamiltonian:
\begin{eqnarray}
	\label{hamdiff}
	\frac{d \boldsymbol{\Theta}}{d \tau} &= &\frac{\partial \mathcal{H}(\boldsymbol{\Theta}, \boldsymbol{\nu})}{\partial \boldsymbol{\nu}}, \nonumber \\
	\frac{d \boldsymbol{\nu}}{d \tau} &=& -\frac{\partial \mathcal{H}(\boldsymbol{\Theta}, \boldsymbol{\nu})}{\partial \boldsymbol{\Theta}}.
\end{eqnarray}

During any time interval $ \Delta \tau$, these equations describe a change from the state at time $\tau$ to that at $\tau + \Delta \tau$. In Bayesian analysis, the parameter $\boldsymbol{\Theta}$ is analogous to positions. The posterior distribution (\ref{post}) can be represented as a canonical distribution by employing a potential energy function as:\\
\begin{equation}
	\mathcal{U}(\boldsymbol{\Theta}) = - log\big(\mathscr{P}(\boldsymbol{\Theta}\mid \mathscr{F}_{t-1}) \big) = -log\big(L(\boldsymbol{\Theta} \mid \mathscr{F}_{t-1})\mathscr{P}_0 (\boldsymbol{\Theta})\big) + \mathscr{C},\nonumber
\end{equation}
where $ \mathscr{C}$ represents a normalizing constant, and the kinetic energy is given by
\begin{equation}
	\mathcal{K}(\boldsymbol{\nu}) = \frac{\boldsymbol{\nu}^T\Sigma^{-1}\boldsymbol{\nu}}{2},
\end{equation}
where the generalized moment $\boldsymbol{\nu}$ will be expressed as a normal random variable $\mathscr{Z} \sim N(\mathbf{0}, \Sigma)$ of the same dimension as $\boldsymbol{\Theta}$, the parameter vector of interest. This yields equations (\ref{ham}) and (\ref{hamdiff}) as
\begin{equation}
	\mathcal{H}(\boldsymbol{\Theta}, \boldsymbol{\nu}) = -log\big(\mathscr{P}(\boldsymbol{\Theta}\mid \mathscr{F}_{t-1}) \big) + \frac{\boldsymbol{\nu}^T\Sigma^{-1}\boldsymbol{\nu}}{2}, \nonumber
\end{equation}
and
\begin{equation}
	\label{hameq}
	\begin{aligned}
		\frac{d \boldsymbol{\Theta}}{d \tau} &= \frac{\partial \mathcal{H}(\boldsymbol{\Theta, \nu)}}{\partial \boldsymbol{\nu}} = \Sigma^{-1}\boldsymbol{\nu},\\\
		\frac{d \boldsymbol{\nu}}{d \tau} &= - \frac{\partial \mathcal{H}(\boldsymbol{\Theta,\nu})}{\partial \boldsymbol{\Theta}}  = \nabla_{\boldsymbol{\Theta}}log\big(\mathscr{P}(\boldsymbol{\Theta}\mid  \mathscr{F}_{t-1} ) \big),
	\end{aligned}
\end{equation}
where $\nabla_{\boldsymbol{\Theta}}log\big(\mathscr{P}(\boldsymbol{\Theta}\mid  \mathscr{F}_{t-1}) \big) = \big(l_{i} + \partial  log\mathscr{P}_0(\boldsymbol{\Theta})/\partial \Theta_{i} \big)$ is the log posterior gradient with $l_{i} =  \partial 	l(\boldsymbol{\Theta}\mid \mathscr{F}_{t-1})/\partial \Theta_{i}$ for $i=1,2,\ldots p+q+2$. In most cases, solving (\ref{hameq}) requires numerical methods, and one such method mentioned earlier is the leapfrog integrator. The  probabilistic programming language, Stan, uses NUTS technique to conduct Bayesian optimization and leverages automatic differentiation techniques to compute gradients of the posterior distribution with respect to model parameters. This feature enables efficient exploration of the parameter space and helps in obtaining more accurate and reliable estimates of posterior quantities.\citep{gelman2015stan}. The following section explains the Monte Carlo simulation study conducted for the model using the two estimation techniques.
\section{Simulation Study}\label{sim}
A simulation study was conducted to evaluate the finite sample performance of the estimators. We consider the following scenarios:
\begin{enumerate}
	\item[(I)]  NoGe-INGARCH(1,1) model with $(\alpha_0, \alpha_1, \beta_1, \phi)^T = (1,0.2, 0.1, 0.05)^T$.
	\begin{itemize}
		\item Underdispersed, low autocorrelation, and low zero - inflation. 
	\end{itemize}
	\item[(II)]  NoGe-INGARCH(1,1) model with $(\alpha_0, \alpha_1, \alpha_2,\phi)^T = (1, 0.3, 0.1, 0.05)^T$.
	\begin{itemize}
		\item Equidispersed, low autocorrelation and low zero-inflation. 
	\end{itemize}
	\item[(III)]  NoGe-INGARCH(1,1) model with $(\alpha_0, \alpha_1, \beta_1,\phi)^T = (1, 0.4, 0.2, 0.55)^T$.
	\begin{itemize}
		\item Overdispersed, high zero - inflation and high autocorrelation.
	\end{itemize}
		\item[(IV)]  NoGe-INGARCH(1,1) model with $(\alpha_0, \alpha_1, \beta_1,\phi)^T = (1, 0.4, 0.2, 0.35)^T$. 
		\begin{itemize}
			\item Overdispersed, moderate zero-inflation and autocorrelation.
		\end{itemize}
\end{enumerate}
The evaluation criteria, in the case of CMLE, are the mean absolute bias (Abs. Bias) and the mean squared error (MSE) defined as:
\begin{equation}
	Abs. Bias(\hat{\Theta})  = \frac{1}{N}\sum_{k=1}^N |\hat{\Theta}_i^{(k)}- \Theta_i^0|,\nonumber
\end{equation}
and
\begin{equation}
	MSE(\hat{\Theta}) = \frac{1}{N}\sum_{k=1}^N (\hat{\Theta}_i^{(k)} - \Theta_i^0)^2 \nonumber,
\end{equation}
respectively, where $N$ is the number of replications, $\hat{\Theta}_i^{(k)}$ is the CMLE corresponding to the $k^{th}$ Monte Carlo replicate, for $i=1,\ldots,p+q+2$ and $\Theta_i^0$ is the true value of the $i^{th}$ parameter. For the present study, we have considered $N=1000$.  Following \cite{andrade2023zero}, the corresponding metrics for evaluating the Bayesian estimates are the square - root of the mean - squared error (RMSE) and mean absolute bias. In the context of Bayesian analysis, $\hat{\Theta}_i^{(k)}$ are the respective posterior mean of the estimate of $\Theta_i^0$ under squared error loss. To obtain Monte Carlo estimates from the posterior means, we generated 25,000 MCMC iterations of the model parameter vector from the complete posterior distribution using the Hamiltonian Monte Carlo (HMC) algorithm.  We discarded the initial $50\%$ of the MCMC iterations as burn-in, resulting in a chain comprising 12,500 posterior samples each of sizes 50, 200 and 500 for each replication. The posterior samples were utilized to compute the Monte Carlo estimates of the posterior mean for each model parameter. The results are exhibited in \textcolor{blue}{Tables} \ref{tab1} to \ref{tab3b} respectively for the various scenarios listed above. Since optimisation of conditional maximum likelihood estimates consumed a lot of time (approximately 12 hrs per simulation for samples sized 50) in R, the simulation exercise was done in MATLAB, whereas the Bayesian analysis was carried out by the R interface of Stan, viz., rstan. 

\begin{table}[H]
	\center
	\caption{Simulation results for (I) }
	\renewcommand{\arraystretch}{1.25}
	\scalebox{0.7}{%
		\begin{tabular}{llrrrrrr}
			\hline
			\multirow{2}{*}{\textbf{\begin{tabular}[c]{@{}l@{}}Sample\\ size\end{tabular}}} & \multicolumn{1}{c}{\multirow{2}{*}{\textbf{Parameter}}} & \multicolumn{3}{c}{\textbf{CMLE}}                                                                                  & \multicolumn{3}{c}{\textbf{Bayesian Estimates}}                                                                     \\ \cline{3-8} 
			& \multicolumn{1}{c}{}                                    & \multicolumn{1}{l}{\textbf{Avg. Est.}} & \multicolumn{1}{l}{\textbf{MSE}} & \multicolumn{1}{l}{\textbf{Abs. Bias}} & \multicolumn{1}{l}{\textbf{Avg. Est.}} & \multicolumn{1}{l}{\textbf{RMSE}} & \multicolumn{1}{l}{\textbf{Abs. Bias}} \\ \hline
			50                                                                              & $\alpha_0=1$                                            & 0.6415                                 & 0.368                            & 0.5116                                 & 1.1005                                 & 0.1006                            & 0.1005                                 \\
			& $\alpha_1=0.2$                                          & 0.4122                                 & 0.2083                           & 0.3324                                 & 0.4321                                 & 0.2336                            & 0.2321                                 \\
			& $\beta_1=0.1$                                           & 0.1747                                 & 0.0913                           & 0.2023                                 & 0.3278                                 & 0.2288                            & 0.2278                                 \\
			& $\phi=0.05$                                             & 0.0501                                 & 0.001                            & 0.0255                                 & 0.0480                                 & 0.0327                            & 0.0269                                 \\ \hline
			200                                                                             & $\alpha_0=1$                                            & 0.8372                                 & 0.1607                           & 0.3035                                 & 1.0999                                 & 0.0999                            & 0.0999                                 \\
			& $\alpha_1=0.2$                                          & 0.3026                                 & 0.092                            & 0.1707                                 & 0.4434                                 & 0.2434                            & 0.2434                                 \\
			& $\beta_1=0.1$                                           & 0.1324                                 & 0.0419                           & 0.1464                                 & 0.3375                                 & 0.2375                            & 0.2375                                 \\
			& $\phi=0.05$                                             & 0.0503                                 & 0.0002                           & 0.0122                                 & 0.0516                                 & 0.0163                            & 0.0134                                 \\ \hline
			500                                                                             & $\alpha_0=1$                                            & 0.9379                                 & 0.0654                           & 0.1906                                 & 1.0995                                 & 0.0995                            & 0.0995                                 \\
			& $\alpha_1=0.2$                                          & 0.2245                                 & 0.0224                           & 0.0676                                 & 0.4438                                 & 0.2439                            & 0.2438                                 \\
			& $\beta_1=0.1$                                           & 0.1258                                 & 0.0251                           & 0.1201                                 & 0.3371                                 & 0.2372                            & 0.2371                                 \\
			& $\phi=0.05$                                             & 0.0501                                 & 0.0001                           & 0.0079                                 & 0.0506                                 & 0.0091                            & 0.0071                                 \\ \hline
	\end{tabular}}
	\label{tab1}
\end{table}

The simulation study reveals how varying levels of dispersion, autocorrelation, and zero-inflation influence the estimation performance under both conditional maximum likelihood estimation (CMLE) and Bayesian inference via Hamiltonian Monte Carlo. From  \textcolor{blue}{Tables} \ref{tab1}, \ref{tab2}, \ref{tab3} and \ref{tab3b}, it is clear that the absolute biases, MSE and RMSE of the estimates decrease with increase in sample size. In Scenario I, characterized by underdispersion, low autocorrelation, and minimal zero-inflation, both CMLE and Bayesian methods perform reasonably well, with decreasing absolute bias and error metrics as the sample size increases. Scenario II, which is equidispersed but maintains low levels of autocorrelation and zero-inflation, also exhibits a similar trend, though the bias and MSE remain slightly higher compared to Scenario I, particularly for CMLE in small samples. Scenarios III and IV introduce more challenging conditions with overdispersion, higher autocorrelation, and significant zero-inflation. In these cases, CMLE performance deteriorates, especially in small samples, with higher biases and MSEs. Conversely, Bayesian estimates remain relatively stable and exhibit better robustness, evidenced by consistently lower RMSE and absolute bias across sample sizes.\par
 In the case of Bayesian estimation by  HMC algorithm, the convergence of Markov chains to the stationary distributions of interest are checked using histograms of the posterior distributions, traceplots and autocorrelation of parameters for the three cases and are given in \textcolor{blue}{Appendix} \ref{plot}. 
Autocorrelation function (ACF) plots visualize the autocorrelation of MCMC chains. High autocorrelation indicates poor mixing, while low autocorrelation suggests better convergence. The ACF plots in \textcolor{blue}{Figures} \ref{fig:3}, \ref{fig:6}, \ref{fig:9} and \ref{fig:12} show that the chains behave well, as autocorrelation quickly drops to zero with increasing lag.
\begin{table}[htp]
	\center
	\caption{Simulation results for (II) }
	\renewcommand{\arraystretch}{1.25}
	\scalebox{0.7}{%
		\begin{tabular}{llrrrrrr}
			\hline
			\multirow{2}{*}{\textbf{\begin{tabular}[c]{@{}l@{}}Sample\\ size\end{tabular}}} & \multicolumn{1}{c}{\multirow{2}{*}{\textbf{Parameter}}} & \multicolumn{3}{c}{\textbf{CMLE}}                                                                                  & \multicolumn{3}{c}{\textbf{Bayesian Estimates}}                                                                     \\ \cline{3-8} 
			& \multicolumn{1}{c}{}                                    & \multicolumn{1}{l}{\textbf{Avg. Est.}} & \multicolumn{1}{l}{\textbf{MSE}} & \multicolumn{1}{l}{\textbf{Abs. Bias}} & \multicolumn{1}{l}{\textbf{Avg. Est.}} & \multicolumn{1}{l}{\textbf{RMSE}} & \multicolumn{1}{l}{\textbf{Abs. Bias}} \\ \hline
			50                                                                              & $\alpha_0=1$                                            & 0.7449                                 & 0.2991                           & 0.465                                  & 1.1005                                 & 0.1006                            & 0.1005                                 \\
			& $\alpha_1=0.3$                                          & 0.4573                                 & 0.1619                           & 0.3097                                 & 0.4321                                 & 0.2336                            & 0.2321                                 \\
			& $\beta_1=0.1$                                           & 0.1310                                 & 0.0546                           & 0.1606                                 & 0.3278                                 & 0.2288                            & 0.2278                                 \\
			& $\phi=0.05$                                             & 0.0501                                 & 0.0010                           & 0.0255                                 & 0.0480                                 & 0.0327                            & 0.0269                                 \\ \hline
			200                                                                             & $\alpha_0=1$                                            & 0.8960                                 & 0.1184                           & 0.2641                                 & 1.0999                                 & 0.0999                            & 0.0999                                 \\
			& $\alpha_1=0.3$                                          & 0.3729                                 & 0.0631                           & 0.1477                                 & 0.4434                                 & 0.2434                            & 0.2434                                 \\
			& $\beta_1=0.1$                                           & 0.1101                                 & 0.0235                           & 0.1153                                 & 0.3375                                 & 0.2375                            & 0.2375                                 \\
			& $\phi=0.05$                                             & 0.0503                                 & 0.0002                           & 0.0122                                 & 0.0516                                 & 0.0163                            & 0.0134                                 \\ \hline
			500                                                                             & $\alpha_0=1$                                            & 0.9736                                 & 0.0427                           & 0.1585                                 & 1.0995                                 & 0.0995                            & 0.0995                                 \\
			& $\alpha_1=0.3$                                          & 0.3155                                 & 0.0132                           & 0.0596                                 & 0.4438                                 & 0.2439                            & 0.2438                                 \\
			& $\beta_1=0.1$                                           & 0.1060                                 & 0.013                            & 0.0899                                 & 0.3371                                 & 0.2372                            & 0.2371                                 \\
			& $\phi=0.05$                                             & 0.0501                                 & 0.0001                           & 0.0079                                 & 0.0506                                 & 0.0091                            & 0.0071                                 \\ \hline
	\end{tabular}}
	\label{tab2}
\end{table}
\begin{table}[H]
	\center
	\caption{Simulation results for (III) }
	\renewcommand{\arraystretch}{1.25}
	\scalebox{0.7}{%
	\begin{tabular}{llrrrrrr}
		\hline
		\multirow{2}{*}{\textbf{\begin{tabular}[c]{@{}l@{}}Sample\\ size\end{tabular}}} & \multicolumn{1}{c}{\multirow{2}{*}{\textbf{Parameter}}} & \multicolumn{3}{c}{\textbf{CMLE}}                                                                                  & \multicolumn{3}{c}{\textbf{Bayesian Estimates}}                                                                     \\ \cline{3-8} 
		& \multicolumn{1}{c}{}                                    & \multicolumn{1}{l}{\textbf{Avg. Est.}} & \multicolumn{1}{l}{\textbf{MSE}} & \multicolumn{1}{l}{\textbf{Abs. Bias}} & \multicolumn{1}{l}{\textbf{Avg. Est.}} & \multicolumn{1}{l}{\textbf{RMSE}} & \multicolumn{1}{l}{\textbf{Abs. Bias}} \\ \hline
		50                                                                              & $\alpha_0=1$                                            & 0.8675                                 & 0.2341                           & 0.4018                                 & 1.0998                                 & 0.0999                            & 0.0998                                 \\
		& $\alpha_1=0.4$                                          & 0.4056                                 & 0.0828                           & 0.2412                                 & 0.4973                                 & 0.0975                            & 0.0973                                 \\
		& $\beta_1=0.2$                                           & 0.2651                                 & 0.0953                           & 0.2469                                 & 0.3099                                 & 0.1099                            & 0.1099                                 \\
		& $\phi=0.55$                                             & 0.5471                                 & 0.0049                           & 0.0563                                 & 0.5444                                 & 0.0645                            & 0.0513                                 \\ \hline
		200                                                                             & $\alpha_0=1$                                            & 0.9805                                 & 0.0944                           & 0.2479                                 & 1.1001                                 & 0.1001                            & 0.1001                                 \\
		& $\alpha_1=0.4$                                          & 0.4038                                 & 0.0256                           & 0.1281                                 & 0.4976                                 & 0.0977                            & 0.0976                                 \\
		& $\beta_1=0.2$                                           & 0.2141                                 & 0.0434                           & 0.1652                                 & 0.3104                                 & 0.1104                            & 0.1104                                 \\
		& $\phi=0.55$                                             & 0.5492                                 & 0.0012                           & 0.0282                                 & 0.5509                                 & 0.0361                            & 0.0285                                 \\ \hline
		500                                                                             & $\alpha_0=1$                                            & 0.9998                                 & 0.0459                           & 0.1711                                 & 1.1001                                 & 0.1002                            & 0.1001                                 \\
		& $\alpha_1=0.4$                                          & 0.4028                                 & 0.0100                           & 0.0795                                 & 0.4981                                 & 0.0982                            & 0.0981                                 \\
		& $\beta_1=0.2$                                           & 0.2003                                 & 0.0205                           & 0.1129                                 & 0.3106                                 & 0.1107                            & 0.1106                                 \\
		& $\phi=0.55$                                             & 0.5496                                 & 0.0005                           & 0.0174                                 & 0.5496                                 & 0.0225                            & 0.0178                                 \\ \hline
	\end{tabular}}
	\label{tab3}
\end{table}
Trace plots display parameter values across iterations, showing how the MCMC samples evolve. The trace plots in \textcolor{blue}{Figures} \ref{fig:2}, \ref{fig:5}, \ref{fig:8} and \ref{fig:11} appear random and well-mixed, with values fluctuating around a central point, indicating no anomalies. Posterior histograms illustrate parameter uncertainty in Bayesian models. The histograms in \textcolor{blue}{Figures} \ref{fig:1}, \ref{fig:4}, \ref{fig:7} and  \ref{fig:10} are stable, showing convergence, with peaks around the most probable values and mostly narrow widths, indicating low uncertainty.
\begin{table}[H]
	\center
	\caption{Simulation results for (IV) }
	\renewcommand{\arraystretch}{1.25}
	\scalebox{0.8}{%
	\begin{tabular}{llrrrrrr}
		\hline
		\multirow{2}{*}{\textbf{\begin{tabular}[c]{@{}l@{}}Sample\\ size\end{tabular}}} & \multicolumn{1}{c}{\multirow{2}{*}{\textbf{Parameter}}} & \multicolumn{3}{c}{\textbf{CMLE}}                                                                                  & \multicolumn{3}{c}{\textbf{Bayesian Estimates}}                                                                     \\ \cline{3-8} 
		& \multicolumn{1}{c}{}                                    & \multicolumn{1}{l}{\textbf{Avg. Est.}} & \multicolumn{1}{l}{\textbf{MSE}} & \multicolumn{1}{l}{\textbf{Abs. Bias}} & \multicolumn{1}{l}{\textbf{Avg. Est.}} & \multicolumn{1}{l}{\textbf{RMSE}} & \multicolumn{1}{l}{\textbf{Abs. Bias}} \\ \hline
		50                                                                              & $\alpha_0=1$                                            & 0.9026                                 & 0.2207                           & 0.392                                  & 1.1001                                 & 0.1001                            & 0.1001                                 \\
		& $\alpha_1=0.4$                                          & 0.4163                                 & 0.0637                           & 0.2015                                 & 0.4823                                 & 0.0824                            & 0.0823                                 \\
		& $\beta_1=0.2$                                           & 0.2268                                 & 0.0696                           & 0.2136                                 & 0.2826                                 & 0.0827                            & 0.0826                                 \\
		& $\phi=0.35$                                             & 0.3485                                 & 0.0041                           & 0.0506                                 & 0.3316                                 & 0.0665                            & 0.0536                                 \\ \hline
		200                                                                             & $\alpha_0=1$                                            & 1.0003                                 & 0.0773                           & 0.2225                                 & 1.1001                                 & 0.1001                            & 0.1001                                 \\
		& $\alpha_1=0.4$                                          & 0.3995                                 & 0.0127                           & 0.0908                                 & 0.4840                                 & 0.0841                            & 0.084                                  \\
		& $\beta_1=0.2$                                           & 0.2019                                 & 0.028                            & 0.1347                                 & 0.2837                                 & 0.0838                            & 0.0837                                 \\
		& $\phi=0.35$                                             & 0.3493                                 & 0.0011                           & 0.0259                                 & 0.3506                                 & 0.0342                            & 0.0278                                 \\ \hline
		500                                                                             & $\alpha_0=1$                                            & 1.0177                                 & 0.0344                           & 0.1483                                 & 1.1002                                 & 0.1003                            & 0.1002                                 \\
		& $\alpha_1=0.4$                                          & 0.3999                                 & 0.0049                           & 0.0565                                 & 0.4834                                 & 0.0835                            & 0.0834                                 \\
		& $\beta_1=0.2$                                           & 0.1908                                 & 0.0123                           & 0.0892                                 & 0.2828                                 & 0.0829                            & 0.0828                                 \\
		& $\phi=0.35$                                             & 0.3501                                 & 0.0004                           & 0.0168                                 & 0.3474                                 & 0.0235                            & 0.0186                                 \\ \hline
	\end{tabular}}
	\label{tab3b}
\end{table}

\section{Forecasting using Bayesian Predictive Distribution}

For the purpose of forecasting future values $X_{t+h}$ for $h = 1, \dots, H$, given $\mathscr{F}_{t}$, we need to compute the predictive probability distribution $\mathscr{P}_{\text{NoGe}}(x_{t+h}\mid\mathscr{F}_{t})$. Using Bayesian inference, the forecast $E(X_{t+h} \mid \mathscr{F}_{t}) = \hat{x}_{t+h}$ can be  determined by:
\begin{equation}
	\label{predprob}
	\mathscr{P}_{\text{NoGe}}(x_{t+h} \mid \mathscr{F}_{t}) = \int_{\boldsymbol{\Omega_\Theta}} \mathscr{P}_{\text{NoGe}}(x_{t+h} \mid \boldsymbol{\Theta}, \mathscr{F}_{t}) \mathscr{P}(\boldsymbol{\Theta} \mid \mathscr{F}_t) d\boldsymbol{\Theta}, 
\end{equation}
where $\boldsymbol{\Omega_\Theta}$ is the parameter space, $\mathscr{P}(\boldsymbol{\Theta} \mid \mathscr{F}_t)$ and $\mathscr{P}_{\text{NoGe}}(x_{t+h} \mid \boldsymbol{\Theta}, \mathscr{F}_{t})$ are assumed integrable and the former is the posterior density. 
Using posterior samples $\{ \boldsymbol{\Theta}^{(m)} \}_{m=1}^{M}$ from the HMC algorithm, the predictive distribution in \eqref{predprob} can be approximated by the Monte Carlo estimate:
\begin{equation}
	\label{predapprox}
	\hat{\mathscr{P}}_{\text{NoGe}}(x_{t+h} \mid \mathscr{F}_{t}) = \frac{1}{M} \sum_{m=1}^{M} \mathscr{P}_{\text{NoGe}}(x_{t+h} \mid \boldsymbol{\Theta}^{(m)}, \mathscr{F}_{t}),
\end{equation}
where 
\smaller
\begin{equation}
	\mathscr{P}_{\text{NoGe}}(x_{t+h} \mid \boldsymbol{\Theta}^{(m)}, \mathscr{F}_{t}) = \phi ^{(m)} \delta (x_{t+h}) + (1-\phi^{(m)})^2\left(1-\frac{1-\phi^{(m)}}{\lambda^{(m)}_{t+h}}\right)^{x_{t+h} - 1}
	. \frac{1}{\lambda^{(m)}_{t+h}}\{1-\delta(x_{t+h})\},
\end{equation}
\normalsize
where $\lambda^{(m)}_{t+h} = \alpha_{0}^{(m)} + \sum_{j=1}^{p}\alpha_i^{(m)}\hat{x}_{t+h-j} +\sum_{j=1}^q\beta_j^{(m)}\hat{\lambda}_{t+h-j}. $ For $h \leq j$, $\hat{x}_{t+h-j} = x_{t+h-j}$ and $\hat{\lambda}_{t+h-j} = \lambda_{t+h-j}$. Whereas for $h \geq j$, $\hat{\lambda}_{t+h-j} = \frac{1}{M}\sum_{m=1}^{M}\lambda_{t+h-j}^{(m)}$, and $\hat{x}_{n+h-j}$ can be estimated using 
\begin{equation}
	\hat{x}_{n+h} = \frac{1}{G} \sum_{m=1}^{M} E(X_{t+h} \mid \boldsymbol{\Theta}^{(m)}, \mathscr{F}_{t}),
\end{equation}
which can be computed as the mean of samples generated using the approximate predictive distribution \eqref{predapprox}. To obtain integer - valued coherent forecasts( See \cite{freeland2004forecasting}), one can consider using the conditional median ($\tilde{x}_{t+h}$) in place of the conditional mean, given by $\sum_{x_{t+h} \leq \tilde{x}_{t+h}} \hat{\mathscr{P}}_{\text{NoGe}}(x_{t+h} \mid \mathscr{F}_{t}) \geq \frac{1}{2}.$

\subsection{ Credible intervals for \(x_{t+h}\)}

Credible intervals for \( x_{t+h} \) can be obtained using the \(100\alpha\%\) and \(100(1 - \alpha)\%\) quantiles of the predictive distribution \( \mathscr{P}_{\text{NoGe}}(x_{t+h} \mid \mathscr{F}_{t}) \) (\cite{chen1999monte}). A \(100(1 - \alpha)\%\) Highest Posterior Density (HPD) region for \( y_{n+h} \) is a subset \( \mathcal{C} \in \mathbb{N}_0 \), defined by:
\[
\mathcal{C} = \left\{ x_{t+h} : \mathscr{P}_{\text{NoGe}}(x_{t+h} \mid \mathscr{F}_{t}) \geq \mathcal{M} \right\},
\]
where \( \mathcal{M} \) is the smallest number on the support \( \mathbb{N}_0 \) such that:
\[
\sum_{x_{t+h} \geq \mathcal{M}} \hat{\mathscr{P}}_{\text{NoGe}}(x_{t+h} \mid \mathscr{F}_{t}) \geq 1 - \alpha.
\]

To compute the HPD region, we use the Monte Carlo estimate of the predictive mass function \eqref{predapprox}. An algorithm for computing the 100(1-$\alpha$) credible intervals is detailed in \cite{andrade2023zero}. The \(100\alpha\%\) and \(100(1 - \alpha)\%\) quantiles, denoted respectively by \( x_{t+h,\alpha} \) and \( x_{t+h,(1-\alpha)} \), are computed as:

\[
x_{t+h,\alpha} = \min \left\{ x_{t+h}^{(\kappa)} \; \bigg| \; \sum_{k=1}^{\kappa} \hat{\mathscr{P}}_{\text{NoGe}}(x_{t+h}^{(k)} \mid \mathscr{F}_{t}) \geq \alpha \right\},
\]

\[
x_{t+h,(1-\alpha)} = \min \left\{ x_{t+h}^{(\kappa)} \; \bigg| \; \sum_{k=1}^{\kappa} \hat{\mathscr{P}}_{\text{NoGe}}(x_{t+h}^{(k)} \mid \mathscr{F}_{t}) \geq (1 - \alpha) \right\}.
\]

Therefore, the \(100(1 - \alpha)\%\) credible interval for \( x_{t+h} \) is $
\mathcal{CI}_{(1 - \alpha)} = \left[ x_{t+h,\alpha}, \; x_{t+h,(1 - \alpha)} \right].$
 In what follows, we elucidate the application of the proposed model, estimation and forecasting methods to real datasets. 

\section{Data Analysis}\label{data}
For the real data application, we consider two data sets viz., weekly data on Hepatitis - B cases and a data on the number of transactions of a stock . We fit four models to the data sets --- NoGe-INGARCH, GP-INGARCH, NB-INGARCH and PINGARCH each with $p=1$ and $q=1$. 
The model adequacy criteria are  Akaike information criterion (AIC) (\cite{aic}), conditional predictive ordinate (CPO) (\cite{gelfand1994bayesian}):
	\[
	\widehat{\text{CPO}}_t = \left( \frac{1}{M} \sum_{m=1}^{M} \frac{1}{\mathscr{P}_{\text{NoGe}}(x_t \mid \boldsymbol{\Theta}^{(m)}, \mathcal{F}_t)} \right)^{-1},
	\]
	and the discrete - case approximation of continuous ranked probability score (CRPS) (\cite{gneiting2007strictly}):		

		\[
		CRPS(F, x) = \int_{-\infty}^{\infty} \left( F(y) - \mathbf{1}(y - x) \right)^2 \, dy,
		\]
where $\mathbf{1}$ is the Heaviside step function, $x$ is the actual observation, and $F$ is the cumulative density function associated with the predictive distribution. The integral may be approximated by discrete sums, as follows:		
		\[
		\text{CRPS}(F, x) = \sum_{j=0}^{x} F(y_j)^2 + \sum_{j=x+1}^{n} (F(y_j) - 1)^2,
		\] 
where the index $n$ corresponds to the last non-zero probability value of the distribution. Since the score is defined for a particular observation, in the context of data analysis, we consider the average CRPS value:
		\[
		\widehat{CRPS} = \frac{1}{N} \sum_{i=1}^{N} \text{CRPS}_i ,
		\]
		where $N$ is the total number of observations in the training set of the data. The model with the least average CRPS is considered better as compared to other fitted models. Likewise, as pointed out by \cite{andrade2023zero}, it is numerically more stable to compute $\log(\text{CPO}_t)$ instead of $\text{CPO}_t$. $\log(\widehat{\text{CPO}})$ and so we consider the average $\log(\widehat{\text{CPO}}_t)$. The model with the least -$\log(\widehat{\text{CPO}})$ is considered the best-fitting model. For evaluation of forecast accuracy, we use the predictive root mean squared error (PRMSE) for the predictive means:
			\begin{equation}\nonumber
			PRMSE(h) = \sqrt{\frac{1}{N^{'}} \sum_{t=1}^{N^{'}}\left( X_{(t)} - \hat{X}^{(h)\text{Mean}}_{(t-h)}\right)^2}, \quad t=N+1, N+2,\dots N+N^{'}; 
		\end{equation} 
		where $N^{'}$ is the total number of test data points, $h$ denotes the forecast horizon, and $\hat{X}^{(h)\text{Mean}}_{(t-h)}$ represents the h-step ahead conditional mean forecast. We also consider predictive mean absolute deviation (PMAD) for predictive medians. (See  \cite{andrews2024forecast}):
			\begin{equation}\nonumber
			PMAD(h) = \frac{1}{N^{'}} \sum_{i=1}^{N^{'}}\left|X_{(t)} - \hat{X}^{(h)\text{Med}}_{(t-h)}\right|, \quad t=N+1, N+2,\dots N+N^{'}, \,
		\end{equation} 
		where $N^{'}$  and $h$ are as previously defined, and $\hat{X}^{(h)\text{Med}}_{(t_i-h)}$ denotes the h-step ahead predictive median.  
		
\subsection{Analysis of Hepatitis - B cases}
Hepadnaviruses (Hepatitis- B) viruses, can lead to both temporary and long-term liver infections. Temporary infections can result in severe illness, with about $0.5\%$ leading to fatal, rapid-onset hepatitis. Chronic infections also pose significant risks, as nearly $25\%$ result in incurable liver cancer.(See \cite{doi:10.1128/mmbr.64.1.51-68.2000}).\par
For the present study, the weekly number of hepatitis - B cases reported in the states of Bremen, Hamburg and Saxony-Anhalt from January 2017 to April 2019 is considered. The data, consisting of 137 observations, is available at \href{https://survstat.rki.de}{https://survstat.rki.de}. We consider data upto Week 110 as training and the remaing 27 data points belong to the test set for evaluating the predictive performance of the models. \textcolor{blue}{Figure} \ref{tshep} depicts the timeseries, ACF and partial autocorrelation function (PACF) plots of the data. \textcolor{blue}{Table} \ref{dat1} gives the parameter estimates with their corresponding standard errors  and posterior standard deviations (for Bayesian estimates) in square brackets respectively, along with AIC, CPO and CRPS criteria. From \textcolor{blue}{Table} \ref{dat1}, we can observe that eventhough GP-INGARCH has, by a small margin, the least AIC, it has higher values with respect to CPO and CRPS.
\begin{figure}[H]
	\centering
	\includegraphics[scale=0.4]{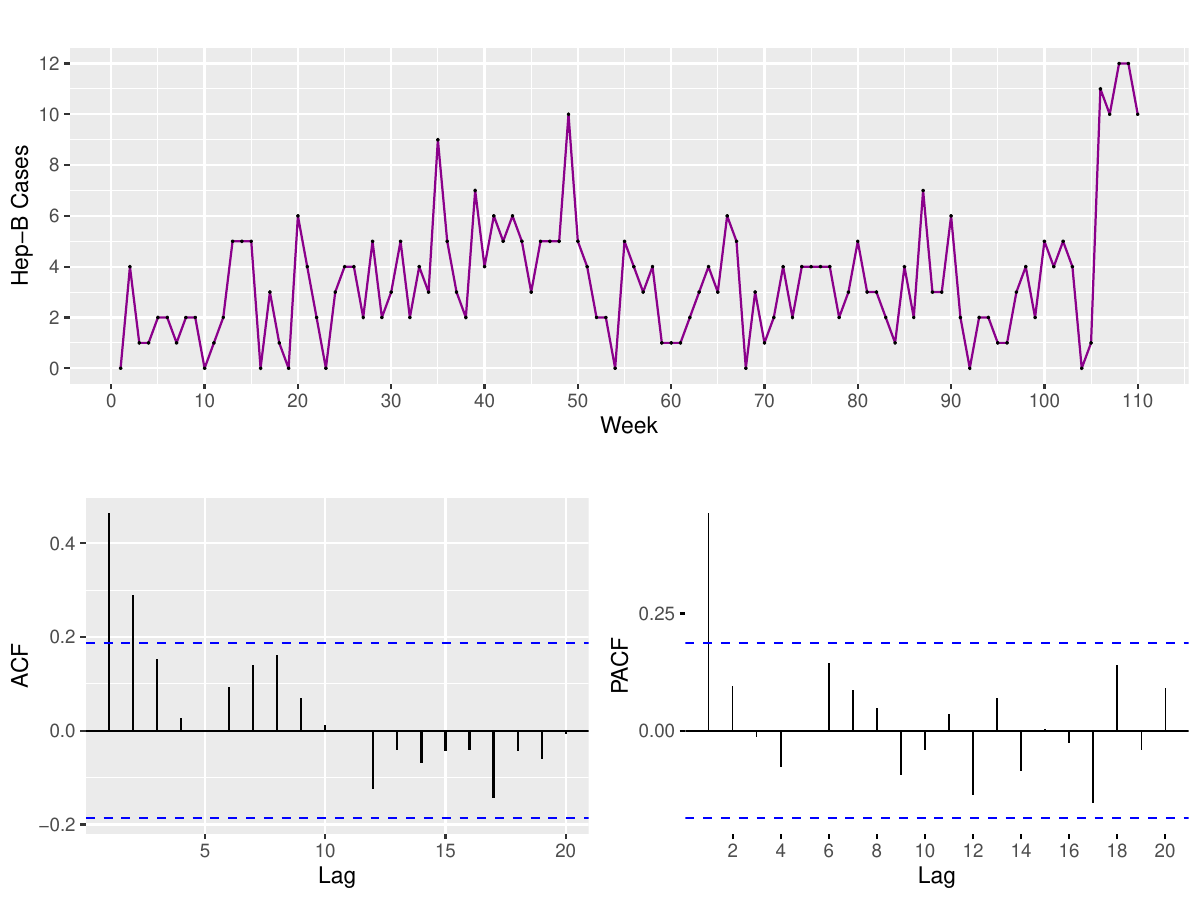}
	\caption{ Time series, ACF and PACF plots of Hepatitis-B data.}
	\label{tshep}
\end{figure}
	The NoGe-INGARCH model performs competitively in terms of AIC, CPO and CRPS criteria, and its ability to accommodate zero inflation is evident from the fitted $\phi$ value. This aligns with the nature of the data, which shows many weeks with zero reported cases. The posterior estimates of autoregressive parameters suggest a mild persistence in infection counts over time, possibly reflecting low but recurring incidence. The incorporation of zero inflation enables a more accurate representation of the data-generating process, which could be influenced by effective vaccination programs or underreporting in some weeks.
\par \textcolor{blue}{Figures} \ref{pithep} and \ref{acfhep} display the PIT histograms and ACF plots of residuals of the models respectively. The PIT histograms conform to uniform distribution behaviour for most cases, and the acf plots are all shown to be stationary with no significant lags greater than zero.
\begin{table}[H]
	\centering
	\caption{CMLEs, Bayesian estimates, and model adequacy criteria of various models fitted to Hepatitis data.}
	\renewcommand{\arraystretch}{1.5}
	\scalebox{0.6}{%
	\begin{tabular}{lccccccccccc}
		\hline
		\multirow{3}{*}{Model}               & \multicolumn{4}{c}{Conditional Maximum Likelihood Estimates} & \multicolumn{4}{c}{Bayesian Estimates}                    & \multirow{3}{*}{AIC} & \multirow{3}{*}{-log($\widehat{CPO}$)} & \multicolumn{1}{l}{\multirow{3}{*}{$\widehat{CRPS}$}} \\ \cline{2-9}
		& \multicolumn{4}{c}{Parameter}                                & \multicolumn{4}{c}{Parameter}                             &                      &                                        & \multicolumn{1}{l}{}                                  \\ \cline{2-9}
		& 1             & 2             & 3             & 4            & 1            & 2            & 3            & 4            &                      &                                        & \multicolumn{1}{l}{}                                  \\ \hline
		NoGe-INGARCH(1,1)                    & 1.5811        & 0.3806        & 0.1787        & 0.7239       & 1.9134       & 0.3089       & 0.1521       & 0.0734       & 469.56               & 250.14                                 & 1.3513                                                \\
		($\alpha_0,\alpha_1,\beta_1,\phi$)   & {[}0.8455{]}  & {[}0.1011{]}  & {[}0.2914{]}  & {[}0.0995{]} & {[}0.0417{]} & {[}0.0082{]} & {[}0.0208{]} & {[}0.0009{]} &                      &                                        &                                                       \\ \hline
		GP-INGARCH(1,1)                      & 1.5750        & 0.3827        & 0.1796        & 0.1522       & 1.8663       & 0.4046       & 0.0762       & 0.1607       & 469.44               & 253.47                                 & 1.3940                                                \\
		($\alpha_0,\alpha_1,\beta_1,\kappa$) & {[}0.2169{]}  & {[}0.0213{]}  & {[}0.2540{]}  & {[}0.0120{]} & {[}0.0352{]} & {[}0.0186{]} & {[}0.0227{]} & {[}0.0052{]} &                      &                                        &                                                       \\ \hline
		NB-INGARCH(1,1)                      & 1.6654        & 0.3471        & 0.1842        & 9.8134       & 1.9095       & 0.3824       & 0.0792       & 9.9657       & 470.86               & 260.14                                 & 1.3760                                                \\
		($\alpha_0,\alpha_1,\beta_1,n$)      & {[}1.0785{]}  & {[}0.1096{]}  & {[}0.3799{]}  & {[}5.2994{]} & {[}0.0718{]} & {[}0.0037{]} & {[}0.0198{]} & {[}0.1832{]} &                      &                                        &                                                       \\ \hline
		PINGARCH(1,1)                        & 1.5251        & 0.3636        & 0.2122        &              & 1.8107       & 0.4055       & 0.0878       &              & 474.36               & 252.36                                 & 1.3587                                                \\
		($\alpha_0,\alpha_1,\beta_1$)        & {[}0.8137{]}  & {[}0.0914{]}  & {[}0.2890{]}  &              & {[}0.0346{]} & {[}0.0038{]} & {[}0.0010{]} &              &                      &                                        & \multicolumn{1}{l}{}                                  \\ \hline
	\end{tabular}}
	\label{dat1}
\end{table}
\begin{table}[H]
	\centering
	\caption{Summary of predictive accuracy measures for the models for the test set - Hepatitis data.}
	\renewcommand{\arraystretch}{2}
	\scalebox{0.5}{%
		\begin{tabular}{lcccc}
			\hline
			\multirow{2}{*}{\begin{tabular}[c]{@{}l@{}}Accuracy\\ Measure\end{tabular}} & \multicolumn{4}{c}{Model}                                                                                                         \\ \cline{2-5} 
			& \multicolumn{1}{l}{PINGARCH} & \multicolumn{1}{l}{NB-INGARCH} & \multicolumn{1}{l}{GP-INGARCH} & \multicolumn{1}{l}{NoGe-INGARCH} \\ \hline
			PRMSE                                                                       & 4.7679                       & 4.4813                         & 4.3841                         &  4.4394                          \\
			PMAD                                                                        & 5.7037                       & 5.3704                         & 4.1481                         & 4.1111                         \\ \hline
	\end{tabular}}
	\label{pred_hep_tab}
\end{table}
\begin{figure}[H]
	\centering
	\includegraphics[scale=0.4]{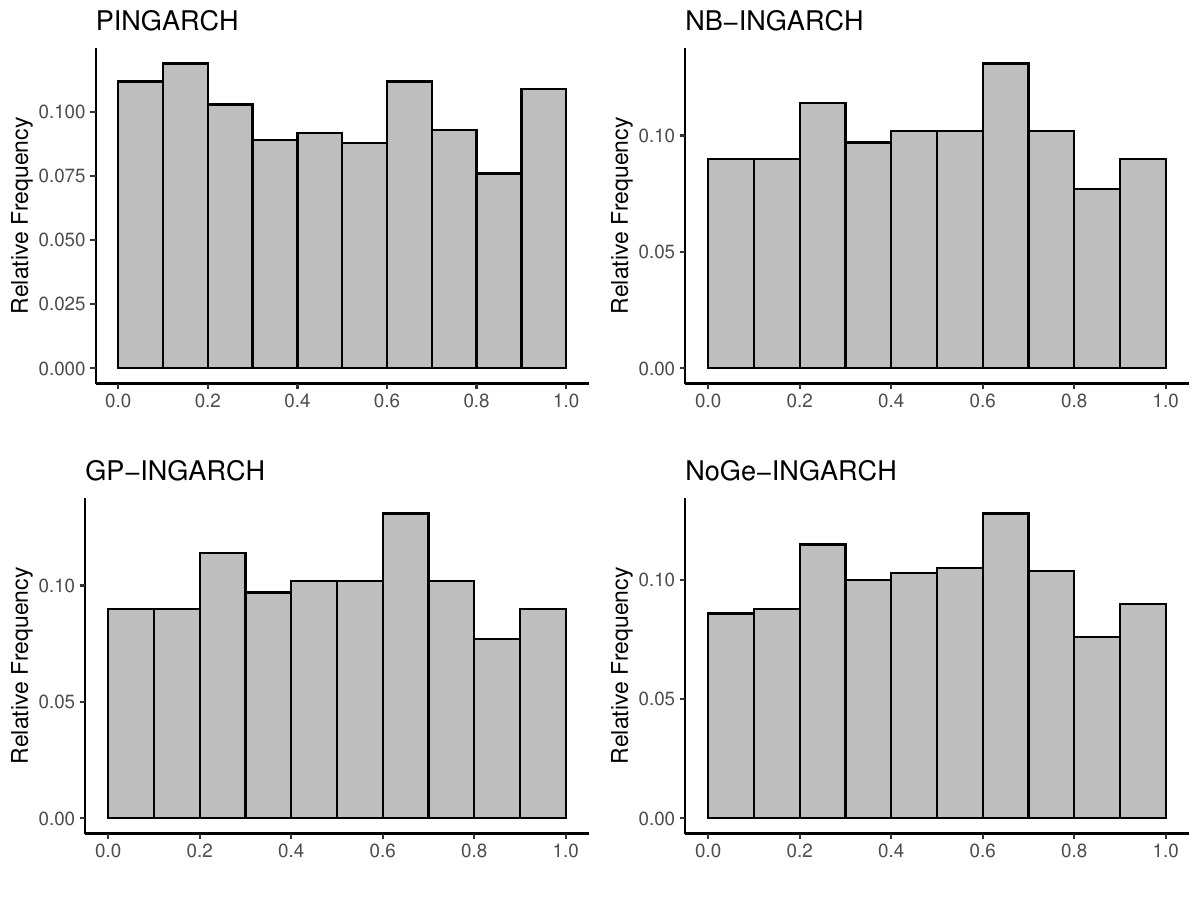}
	\caption{ PIT histograms of models fitted to Hepatitis-B data.}
	\label{pithep}
\end{figure}
Table \ref{pred_hep_tab} enlists the predictive accuracy measures for the one - step ahead mean and median forecasts provided by the models for the test set of Hepatitis - B data. It is evident that the least value of PRMSE corresponds to GP-INGARCH model. However, NoGe-INGARCH possesses least PMAD and provides more precise forecasts as shown in Figure \ref{pred_hep}.
\begin{figure}[H]
	\centering
	\includegraphics[scale=0.4]{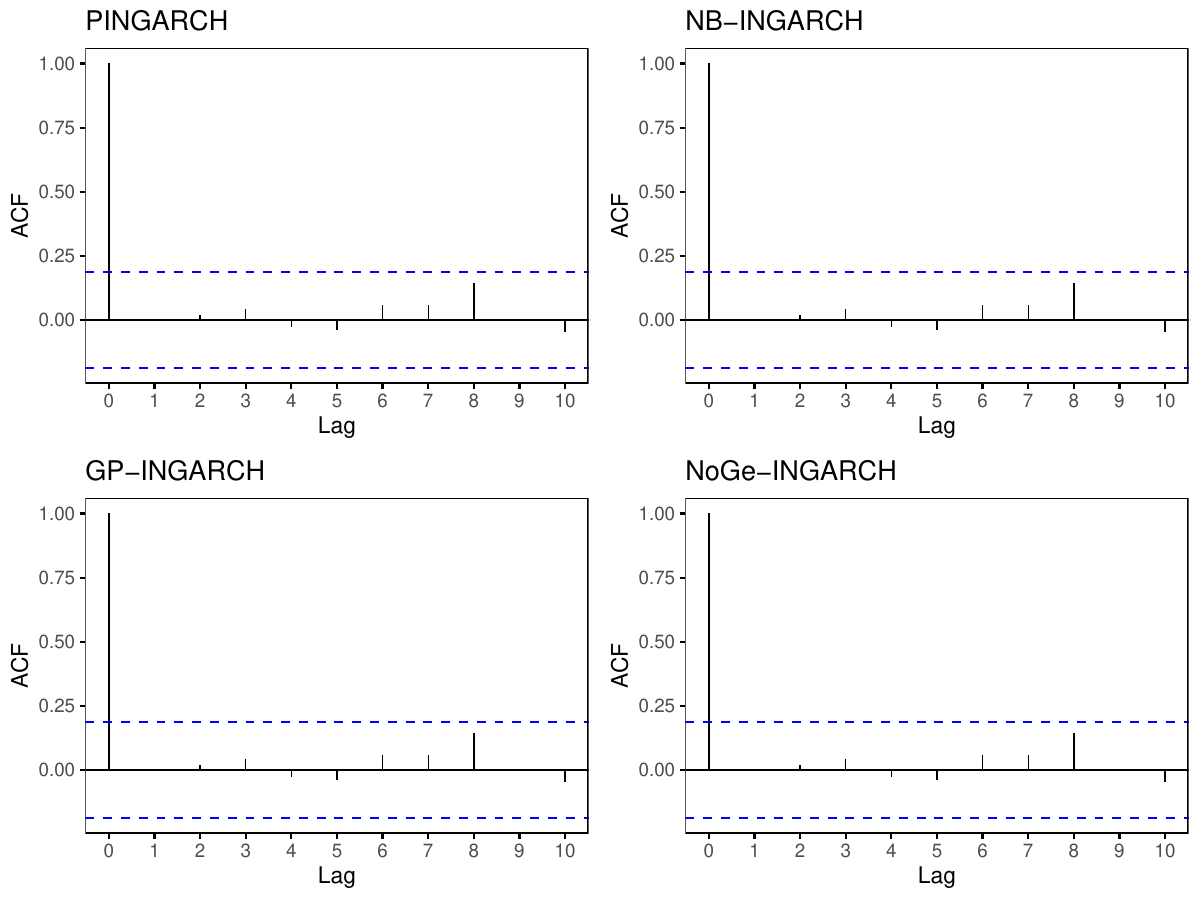}
	\caption{ ACF plots of residuals of models fitted to Hepatitis-B data.}
	\label{acfhep}
\end{figure}
	\begin{figure}[htbp]
	\centering
	\begin{minipage}[b]{0.47\textwidth}
		\centering
		\includegraphics[width=\textwidth]{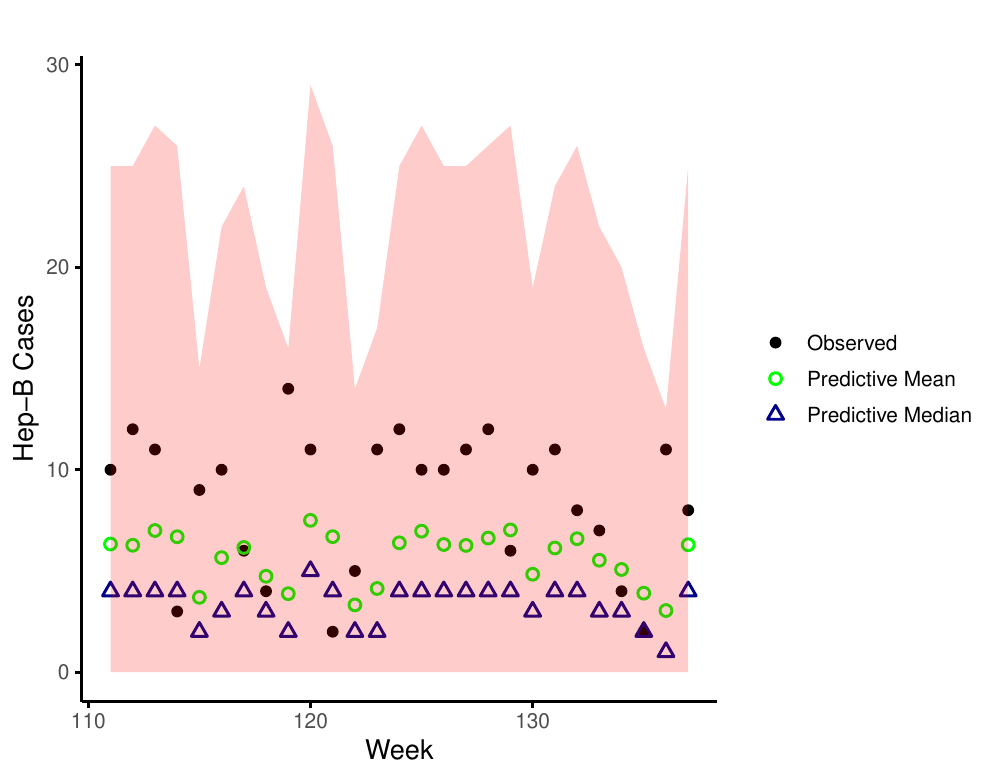}
		\caption*{(a) GP-INGARCH} 
	\end{minipage}
	\hfill
	\begin{minipage}[b]{0.47\textwidth}
		\centering
		\includegraphics[width=\textwidth]{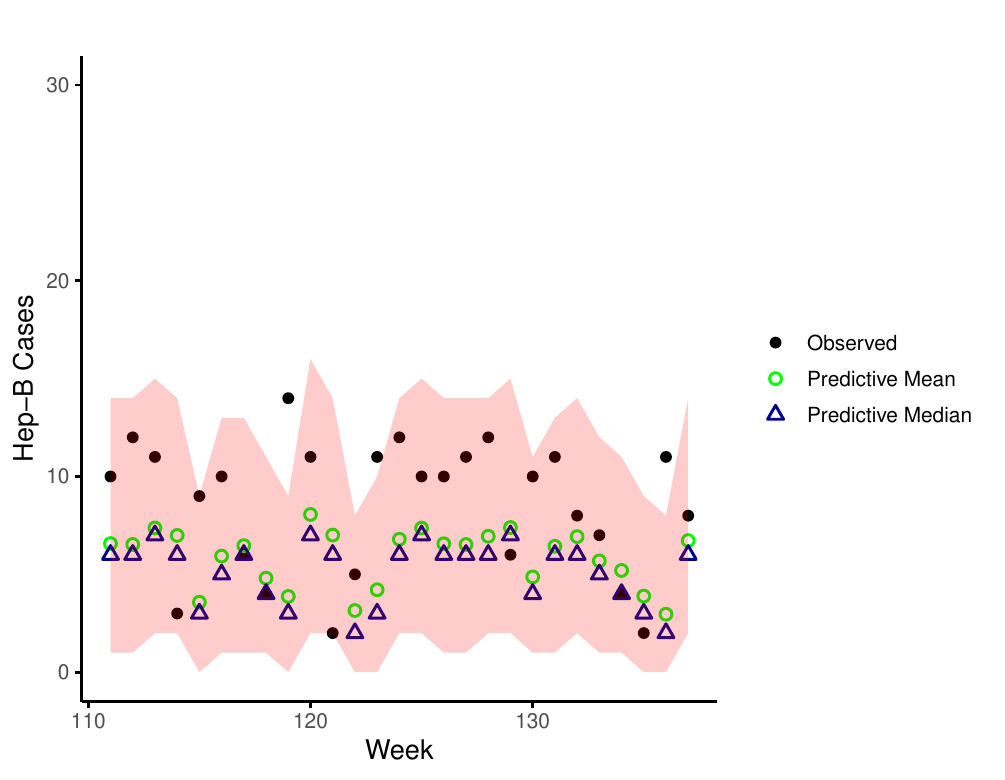}
		\caption*{(b) NoGe-INGARCH}
	\end{minipage}
	\caption{One -step ahead predictions for Weeks 111 to 137 with 95 \% credible intervals  using GP-INGARCH and NoGe-INGARCH models.}
	\label{pred_hep}
\end{figure}
\subsection{Analysis of Transactions data}
The number of transactions of the Ericsson - B stock per minute between 9:35 and 17:14 on 2 July 2002, encompassing 460 observations sourced from \cite{weiss2018introduction}. \textcolor{blue}{Figure} \ref{tstrans} shows the timeseries, ACF and PACF plots corresponding to the count time series. We consider the observations from 9:35 upto 15:42 as the training set and those between 15:42 and 17:15 as the test set.
\begin{figure}[H]
	\centering
	\includegraphics[scale=0.4]{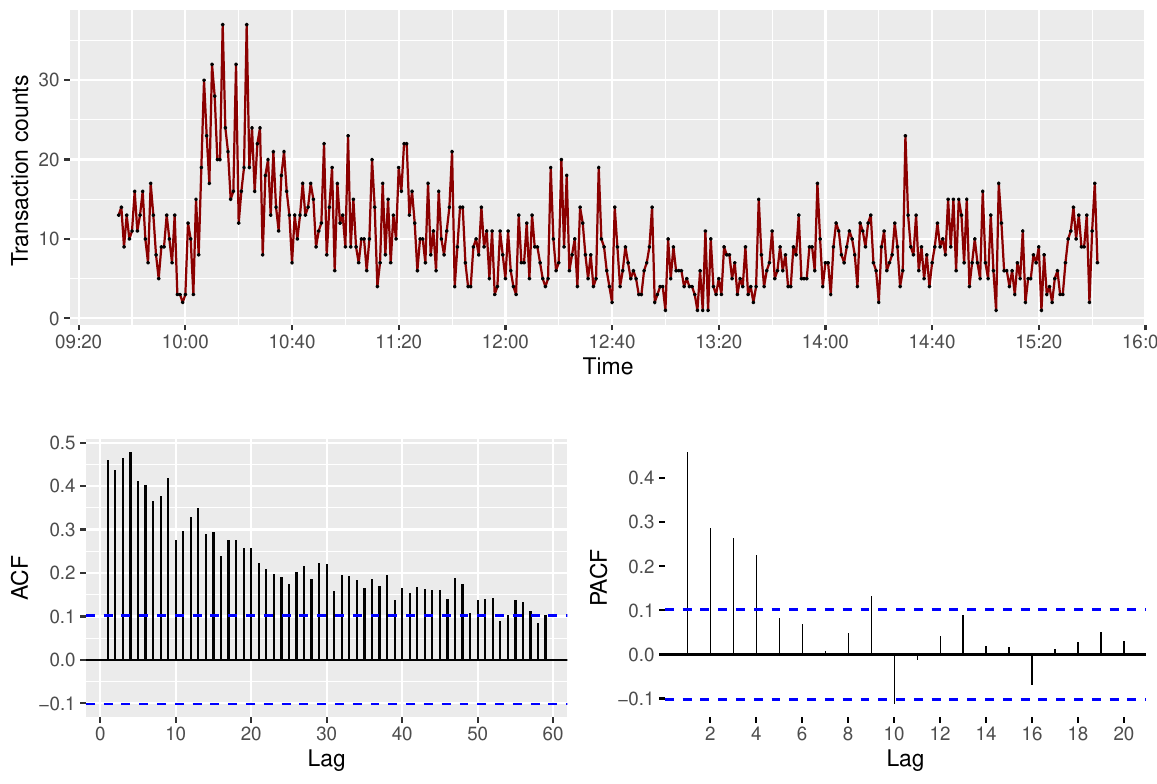}
	\caption{ Time series, ACF and PACF plots of Transactions data.}
	\label{tstrans}
\end{figure}
The data was originally published by \cite{brannas2010integer} and provides the data for working days between 2 and 22 July 2002. \cite{weiss2018introduction} concluded that the GP-INGARCH(1,1) model fits the data well, based on the AIC. \par
\textcolor{blue}{Table} \ref{dat} enlists the parameter estimates with their corresponding standard errors and posterior standard deviations in square brackets. The minimum AIC and CRPS corresponds to NoGe-INGARCH(1,1). In the context of high-frequency stock transaction data, the NoGe-INGARCH model captures the volatility clustering and persistence in transaction counts effectively, as indicated by high values of the autoregressive parameters. The near-zero value of $\phi$ reflects the expected absence of zero inflation, confirming that transactions occurred in nearly every time interval.
\begin{table}[H]
	\centering
	\caption{CMLEs, Bayesian estimates, and model adequacy criteria of various models fitted to Transactions data.}
	\renewcommand{\arraystretch}{1.5}
	\scalebox{0.5}{%
		\begin{tabular}{lccccccccccl}
			\hline
			\multicolumn{1}{c}{}                 & \multicolumn{4}{c}{Conditional Maximum Likelihood Estimates}                              & \multicolumn{4}{c}{Bayesian Estimates}                                                    & \multirow{3}{*}{AIC} & \multirow{3}{*}{-log($\widehat{CPO}$)} & \multirow{3}{*}{$\widehat{CRPS}$} \\ \cline{2-9}
			Model                                & \multicolumn{4}{c}{Parameter}                                                             & \multicolumn{4}{c}{Parameter}                                                             &                      &                                        &                                   \\ \cline{2-9}
			\multicolumn{1}{r}{}                 & 1                    & 2                    & 3                    & 4                    & 1                    & 2                    & 3                    & 4                    &                      &                                        &                                   \\ \hline
			NoGe-INGARCH(1,1)                    & 0.3013               & 0.1613               & 0.8055               & 0.1699               & 0.2924               & 0.1334               & 0.8362               & 0.0001               & 2099.94              & 1124.92                                & \multicolumn{1}{c}{2.5748}        \\
			($\alpha_0, \alpha_1, \beta_1,\phi$) & [0.1522]             & [0.0063]             & [0.0806]             & [0.0199]             & [0.0038]             & [0.0017]             & [0.0011]             & [0.0000]             &                      &                                        &                                   \\ \hline
			GP-INGARCH(1,1)                      & 0.3016               & 0.1597               & 0.8087               & 0.3181               & 0.2621               & 0.1347               & 0.8378               & 0.3229               & 2100.20              & 1051.08                                & \multicolumn{1}{c}{2.5763}        \\
			($\alpha_0,\alpha_1,\beta_1,\kappa$) & [0.1550]             & [0.0315]             & [0.0393]             & [0.0266]             & [0.0028]             & [0.0010]             & [0.0006]             & [0.0006]             &                      &                                        &                                   \\ \hline
			NB-INGARCH(1,1)                      & 0.2842               & 0.1475               & 0.8225               & 8.7250               & 0.2469               & 0.1349               & 0.8392               & 8.7742               & 2105.05              & 1209.96                                & \multicolumn{1}{c}{2.5774}        \\
			($\alpha_0, \alpha_1, \beta_1, n$)   & [0.1510]             & [0.0295]             & [0.0377]             & [1.2600]             & [0.0052]             & [0.0109]             & [0.0007]             & [0.0098]             &                      &                                        &                                   \\ \hline
			PINGARCH(1,1)                        & 0.2991               & 0.1616               & 0.8079               &                      & 0.2547               & 0.1389               & 0.8344               &                      & 2239.00              & 1169.43                                & \multicolumn{1}{c}{2.6086}        \\
			($\alpha_0, \alpha_1, \beta_1$)      & [0.1072]             & [0.0208]             & [0.0264]             &                      & [0.0101]             & [0.0078]             & [0.0090]             &                      &                      &                                        &                                   \\ \hline
	\end{tabular}}
	\label{dat}
\end{table}

\begin{figure}[H]
	\centering
	\includegraphics[scale=0.4]{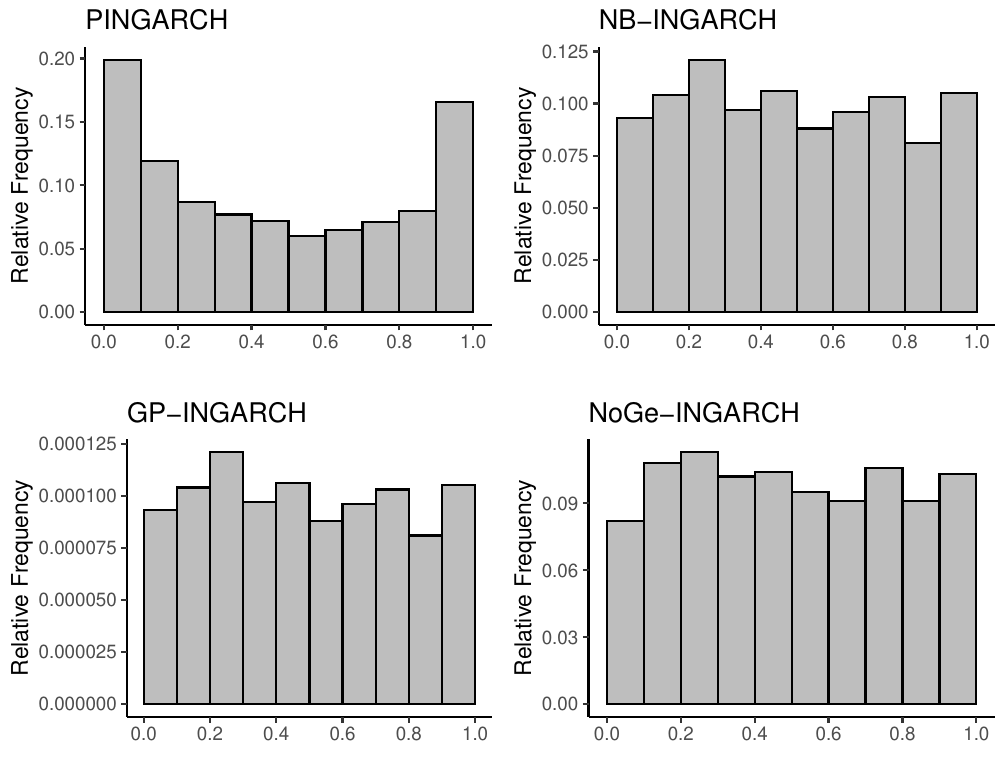}
	\caption{ PIT histograms following analyis of Transactions data.}
	\label{dafig}
\end{figure}
From \textcolor{blue}{Figure} \ref{dafig}, it can be seen that the PIT histograms of NoGe-INGARCH(1,1), GP-INGARCH(1,1) and NB-INGARCH(1,1) conform to the uniform pattern. A conclusion similar to the previous analysis that the residuals are uncorrelated can be arrived at by inspection of \textcolor{blue}{Figure} \ref{acftrans}. Compared to other models, NoGe-INGARCH balances flexibility and parsimony, yielding both low CRPS and good diagnostic behavior. These results demonstrate that the proposed model can adapt to both sparse and dense count processes across diverse application domains.\par
Table \ref{pred_tab_trans} summarizes the predictive performance of the models for the Transactions data. NoGe-INGARCH model possesses the least PRMSE and PMAD confirming better forecasts as compared to the remaining models. Moreover, the credible interval corresponding to NoGe-INGARCH covers most of the observations in the test set as shown in Figure \ref{pred_trans}.
\begin{figure}[H]
	\centering
	\includegraphics[scale=0.4]{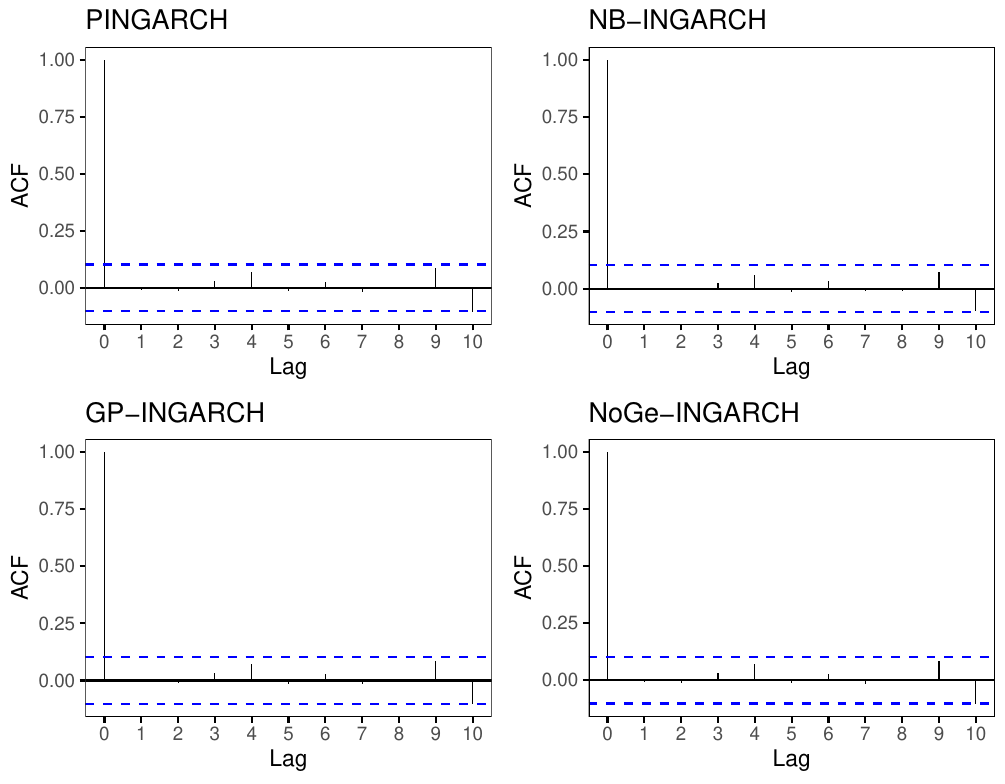}
	\caption{ ACF plots of residuals of models following analyis of Transactions data.}
	\label{acftrans}
\end{figure}
\begin{table}[H]
	\centering
	\caption{Summary of predictive accuracy measures for the models for Transactions data's test set.}
	\renewcommand{\arraystretch}{2}
	\scalebox{0.6}{%
		\begin{tabular}{lcccc}
			\hline
			\multirow{2}{*}{\begin{tabular}[c]{@{}l@{}}Accuracy\\ Measure\end{tabular}} & \multicolumn{4}{c}{Model}                                                                                                         \\ \cline{2-5} 
			& \multicolumn{1}{l}{PINGARCH} & \multicolumn{1}{l}{NB-INGARCH} & \multicolumn{1}{l}{GP-INGARCH} & \multicolumn{1}{l}{NoGe-INGARCH} \\ \hline
			PRMSE                                                                                    & 5.3365                       & 5.3237                         & 5.3306                         & 5.2615                           \\
			PMAD                                                                                     & 4.9021                       & 4.8587                        & 4.0870                          & 3.9730                          \\ \hline
	\end{tabular}}
\label{pred_tab_trans}
\end{table}
	\begin{figure}[htbp]
	\centering
	\begin{minipage}[b]{0.45\textwidth}
		\centering
		\includegraphics[width=\textwidth]{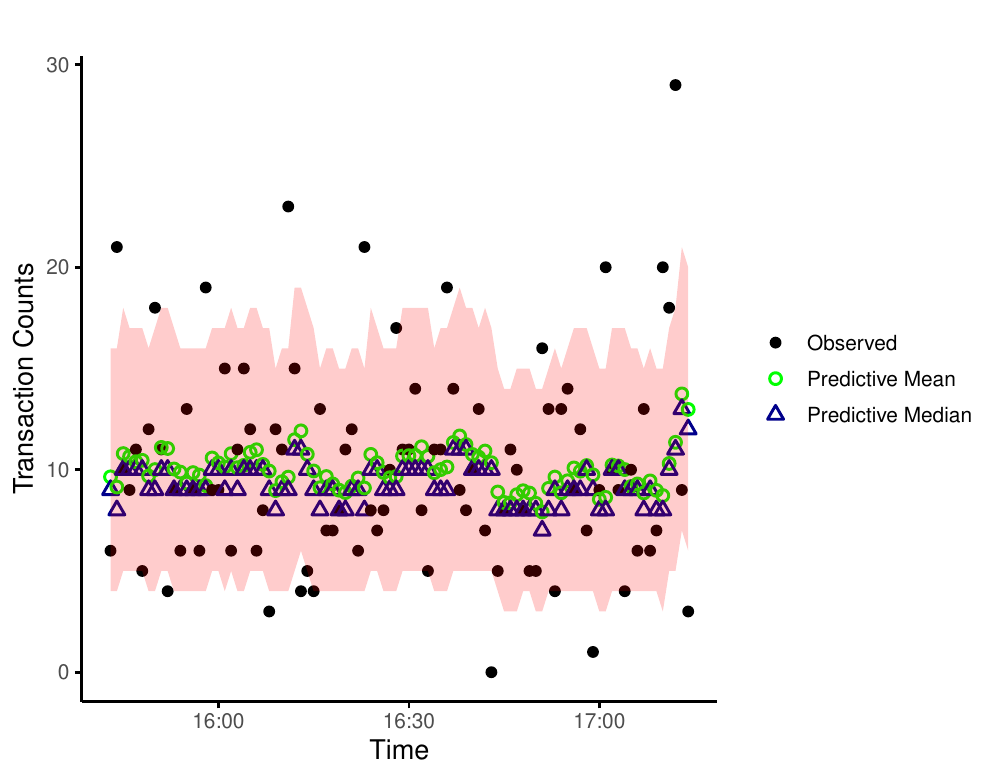}
		\caption*{(a) GP-INGARCH} 
	\end{minipage}
	\hfill
	\begin{minipage}[b]{0.45\textwidth}
		\centering
		\includegraphics[width=\textwidth]{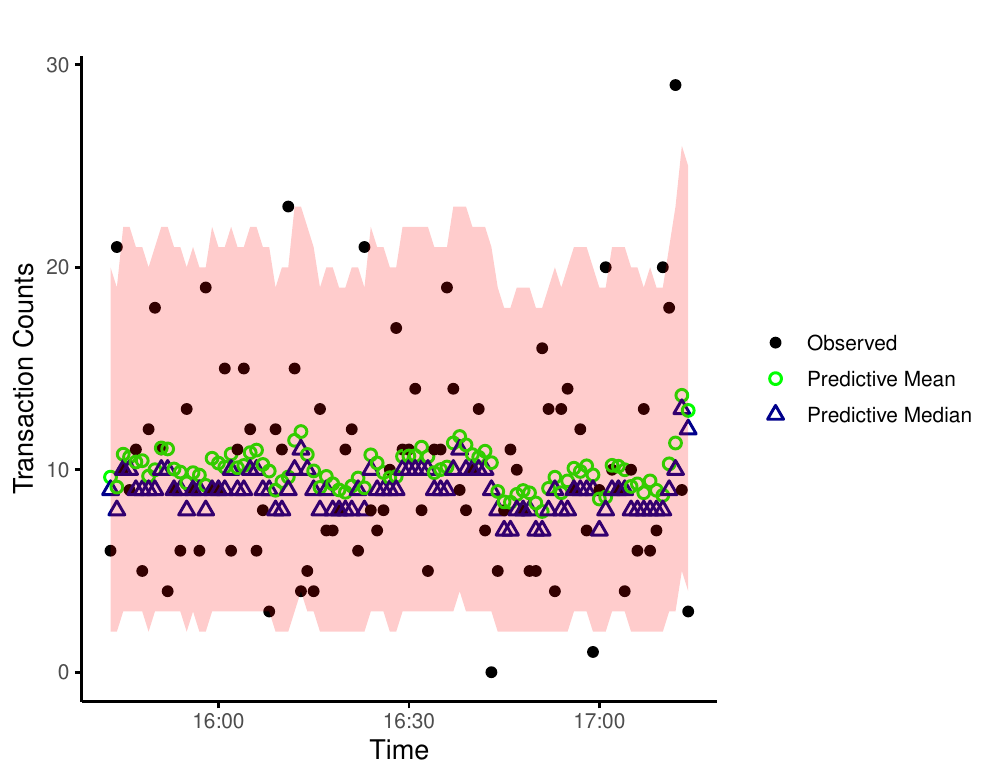}
		\caption*{(b) NoGe-INGARCH}
	\end{minipage}
	\caption{One -step ahead predictions for 15:43 to 17:14 with 95 \% credible intervals  using GP-INGARCH and NoGe-INGARCH models.}
	\label{pred_trans}
\end{figure}

\section{Conclusion}\label{conc}
In this paper, we introduced a novel geometric INGARCH (NoGe-INGARCH) process and discussed some statistical properties. The methods of estimation by conditional maximum likelihood and Bayesian approach via the HMC algorithm in the context of INGARCH models were detailed. A simulation study was conducted to compare the estimation procedures in the case of the NoGe-INGARCH model. The performance of the NoGe-INGARCH model was then assessed on two datasets and compared with existing INGARCH-type models in terms of goodness of fit and predictive accuracy.

\section*{Declaration of interest}
\noindent  No potential conflict of interest was reported by the authors.

\section*{Acknowledgement}
The authors wish to express their gratitude to the referees whose constructive suggestions have led to the present, refined version of the article. 

\section*{Data Availability Statement}
Data sharing is not applicable to this article as no new data were created or analyzed in this study.

\bibliography{sample}

\bibliographystyle{apalike}
\begin{appendices}
	\section{Appendix} \label{append}
	In this Appendix, we lay out the proofs of \textcolor{blue}{Theorems} \ref{thm2}  and \ref{thm3} stated in \textcolor{blue}{Section} \ref{nogeingarch} and plots of convergence diagnostics for the Bayesian estimation of parameters in the simulation study conducted in \textcolor{blue}{Section} \ref{sim}. 
	\subsection{Proof of \textcolor{blue}{Theorem} \ref{thm2}}
	\begin{proof} \label{pf1}
		The proof is influenced by the works of \cite{fong2007mixture} and \cite{zhu2011negative}. Assuming first-order stationarity, let 
		\begin{equation}
			\begin{aligned}
				\gamma_{it} &= Cov[X_t, X_{t-i}],\\  \nonumber
				&= E[X_tX_{t-i}] - E[X_t]E[X_{t-i}], \\
				&=E[X_tX_{t-i}] -\mu^2,\; i=0,1,\ldots,p.
			\end{aligned}
		\end{equation}
		So, we need only consider $E[X_tX_{t-i}]$ to derive the conditions as $\mu^2$ remains a constant. Additionally, assume $\mathcal{C}$ to be a constant independent of $t$. If the process is second - order stationary, we have 
		\begin{equation}
			\gamma_{jt} = \gamma_{j,t-i}, \quad i=0,1,\ldots p, \; \text{and}\; j \in \mathbb{Z},\nonumber
		\end{equation}
		where $\mathbb{Z} = \{\ldots,-1,0,1,\ldots\}$. From (\ref{eqr}), and by definition we arrive at 
		\small
		\begin{equation}\nonumber
			\begin{aligned}
				E[X_{t}^2|\mathscr{F}_{t-1}] &= Var[X_{t}|\mathscr{F}_{t-1}]+ E^2[X_{t}|\mathscr{F}_{t-1}] = \lambda_t \left(\frac{1+\phi}{1-\phi} \lambda_t - 1 \right) + \lambda_t^2\\
				&=  \alpha_0(\zeta \alpha_0 -1) + (2 \zeta \alpha_0-1)\sum_{i=1}^p\alpha_iX_{t-i} + \zeta\Big\{ \sum_{i=1}^p \alpha_i^2 X_{t-i}^2 + \sum_{i\neq j}^p \alpha_i\alpha_j X_{t-i} X_{t-j} \Big\},
			\end{aligned}
		\end{equation}
		\normalsize
		where $\zeta = \frac{2}{1-\phi}$. For $s= 1, \ldots,p-1$,
		\small
		\begin{equation}
			\begin{aligned}
				\gamma_{st} + \mu ^2  &= E\big[E[X_t|\mathscr{F}_{t-1}]X_{t-s}\big]  \nonumber\\ 
				&= E\big[ \alpha_0 X_{t-s} + \sum_{i=1}^p \alpha_iX_{t-i}X_{t-s}\big]  = \alpha_0\mu + \sum_{\substack{i=1 \\ i \neq s}}^p \alpha_i \left(\gamma_{|i-s|,t} + \mu^2\right) + \alpha_s\left(\gamma_{0,t-s} + \mu^2 \right)\\
				&= \alpha_0\mu + \alpha_s \left(\gamma_{0,t-s} + \mu^2\right)+ \sum_{|i-s|=1}\alpha_i\left(\gamma_{1t}+ \mu^2\right)\ldots  +\sum_{|i-s|=p-1}\alpha_i\left(\gamma_{p-1,t} +\mu^2\right),
			\end{aligned}
		\end{equation}
		where $\gamma_{j,t-i}$ are replaced by $\gamma_{j,t}$ for $i=1,\ldots,p-1$ and $j \in \mathbb{Z}$. Hence, for $s= 1, \ldots,p-1$,
		\begin{equation}
			\alpha_0\mu + \nu_{s0}\left(\gamma_{0,t-s}+\mu^2\right) + \sum_{r=s}^{p-1} \nu_{sr}\left(\gamma_{rt} + \mu^2\right)=0. \nonumber
		\end{equation}
		\normalsize
		Therefore, 
		\begin{equation}
			M\big(\gamma_{1t} + \mu^2,\ldots,\gamma_{p-1,t} +\mu^2\big)^T = - \Big( \alpha_0\mu 
			+ \nu_{10}\left(\gamma_{0,t-1} +\mu^2\right) ,\ldots,\alpha_0\mu + \nu_{p-1,0}\left( \gamma_{0,t-p+1}+\mu^2\right)\Big)^T, \nonumber
		\end{equation}
		then
		\begin{equation}
			\gamma_{st} + \mu^2= -\alpha_0\mu \sum_{r=1}^{p-1}m_{sr} - \sum_{r=1}^{p-1}m_{sr}\nu_{r0}\left(\gamma_{0,t-r}+\mu^2\right), \quad s=1,\ldots ,p-1. \nonumber
		\end{equation}
		Let $\mathcal{C}= \zeta\alpha_0^2 - \alpha_0(1+\mu) + (2\zeta\alpha_0 -1)\sum_{i=1}^p\alpha_i\mu $ . The unconditional second moment of $X_t$, $\gamma_{0t} =E[X_t^2]$, can be rewritten as
		\small
		\begin{equation}
			\begin{aligned}
				\gamma_{0t} &= \mathcal{C}+ \zeta \left\{\sum_{u=1}^p\alpha_r^2\gamma_{0,t-r} + \sum_{i,j=1}^p \alpha_i\alpha_j\gamma_{|i-j|,t}\right\} = \mathcal{C}+ \zeta\left\{ \sum_{r=1}^p\alpha_r^2\gamma_{0,t-r} +\sum_{v=1}^{p-1}\sum_{|i-j|=v} \alpha_i\alpha_j\gamma_{vt}\right\}\\ \nonumber
				&= \mathcal{C}+ \zeta \left\{ \sum_{r=1}^p\alpha_r^2\gamma_{0,t-r} + \sum_{v=1}^{p-1}\sum_{|i-j|=v} \alpha_i\alpha_j \big(-\alpha_0\mu\sum_{u=1}^{p-1}m_{sr} - \sum_{r=1}^{p-1}m_{sr}\nu_{r0}\gamma_{0,t-r}\big)\right\}\\
				&= \mathcal{C}_0 + \zeta\left\{\sum_{r=1}^p\alpha_r^2\gamma_{0,t-r} - \sum_{r=1}^{p-1}\big(\sum_{v=1}^{p-1}\sum_{|i-j|=v} \alpha_i\alpha_j m_{vr} \nu_{r0}\big)\gamma_{0,t-r}\right\}\\
				&= \mathcal{C}_0 + \zeta \left\{\sum_{r=1}^{p-1}\left(\alpha_r^2 - \sum_{v=1}^{p-1}\sum_{|i-j|=v} \alpha_i\alpha_j m_{vr} \nu_{r0}\right)\gamma_{0,t-r}+ \alpha_p^2\gamma_{0,t-p}\right\},
			\end{aligned}
		\end{equation}
		\normalsize
		or equivalently, 
		\begin{equation}
			\gamma_{0t} = \mathcal{C}_0 + \sum_{r=1}^p L_r \gamma_{0,t-r}, \nonumber
		\end{equation}
		where $L_r = \zeta\big(\alpha_r^2 - \sum_{v=1}^{p-1}\sum_{|i-j|=v}\alpha_i\alpha_jm_{vr}\nu_{r0}\big), u=1,\ldots,p-1$,  $L_p = \zeta\alpha_p^2$ and $\mathcal{C}_0 = \mathcal{C} - \alpha_0\mu\sum_{v=1}^{p-1}\alpha_i\alpha_j\sum_{r=1}^{p-1}m_{vr}$. Therefore, the non-homogeneous difference equation has a stable solution if the equation $1-L_1b^{-1} - \ldots - L_p b^{-p} =0$ has all roots inside the unit circle. 
	\end{proof}
	\subsection{Proof of \textcolor{blue}{Theorem} \ref{thm3}} \label{pf}
	\begin{proof}
		Let $I_t$ be the $\sigma-$ field generated by $\{\lambda_t, \lambda_{t-1},\ldots,\}$, then we have
		\begin{equation}
			E[X_t\mid\mathscr{F}_{t-1},I_t] = 	E[X_t\mid\mathscr{F}_{t-1}] = \lambda_t.
		\end{equation}
		From (\ref{eqr}), we have $E[X_t\mid\mathscr{F}_{t-1}] = \lambda_t ,$ and $Var[X_t\mid \mathscr{F}_{t-1}]  = \lambda_t\left(\frac{1+\phi}{1-\phi}\lambda_t -1\right)$. So, for $h\geq 0$
		\begin{equation}
			\label{eq6}
			\begin{aligned}
				Cov[X_t - \lambda_t, \lambda_{t-h}] &= E\Big[\big(X_t-\lambda_t - \underbrace{E[X_t-\lambda_t]}_0\big)\big(\lambda_{t-h} - \underbrace{E[\lambda_{t-h}}_\mu]\big)\Big]\\
				&=   E\Big[\big(X_t-\lambda_t\big)\big(\lambda_{t-h} - \mu\big)\Big]\\
				&= E\Big[\big(\lambda_{t-h} - \mu\big) E\big[X_t-\lambda_t \mid I_t \big]\Big]\\
				&= E\Big[\big(\lambda_{t-h} - \mu\big) \big(E\big[E[X_t \mid\mathscr{F}_{t-1},I_t]\mid I_t \big]-\lambda_t\big) \Big]\\
				&= E\Big[\big(\lambda_{t-h} - \mu\big)\big( E\big[\lambda_t \mid I_t \big] - \lambda_t\big)\Big]= 0.\\
			\end{aligned}
		\end{equation} 
		
		Similarly, for $h<0$, from (\ref{eq2}), we have 
		\begin{equation}
			\label{eq7b}
			\begin{aligned}
				Cov[X_t , X_{t-h} -\lambda_{t-h}] &=E\Big[\big(X_t-\mu\big)\big(X_{t-h} - \lambda_{t-h}\big)\Big]\\
				&= E\Big[\big(X_t-\mu\big)E\big[\big(X_{t-h} - \lambda_{t-h}\big)\mid \mathscr{F}_{t-h-1}\big]\Big]\\
				&= E\Big[\big(X_t-\mu\big)\underbrace{\big(\lambda_{t-h}-E\big[ \lambda_{t-h}\mid \mathscr{F}_{t-h-1}\big]\big)}_0 \Big]=0.
			\end{aligned}
		\end{equation} 
		From (\ref{eq6}) and (\ref{eq7b}), we have 
		\begin{equation}
			\label{eq8}
			Cov[X_t,\lambda_{t-h}] =\begin{cases} Cov[\lambda_t, \lambda_{t-h}]; & h\geq0 ,\\\\
				Cov[X_t,X_{t-h}]; & h<0,
			\end{cases}
		\end{equation}
		which is due to the property $Cov[X-Y,Z] = Cov[X,Z] - Cov[Y,Z]$.
		For $h \geq 0$, from (\ref{eq2}) and (\ref{eq8}), we have
		\smaller
		\begin{equation}
			\label{eq8b}
			\begin{aligned}
				\gamma_{\lambda}(h) &= Cov[\lambda_t,\lambda_{t-h}]\\
				&= \sum_{i=1}^p\alpha_iCov[X_{t-i},\lambda_{t-h}] + \sum_{j=1}^q \beta_j Cov[\lambda_{t-j},\lambda_{t-h}]\\
				&= \sum_{i=1}^{min(h,p)}\alpha_iCov[\lambda_{t-i},\lambda_{t-h}]+\sum_{i=h+1}^p\alpha_iCov[X_{t-i},X_{t-h}] + \sum_{j=1}^q\beta_jCov[\lambda_{t-j},\lambda_{t-h}].
			\end{aligned}
		\end{equation} 
		\normalsize
		Similarly, for $h \geq 1$, we have
		\smaller
		\begin{equation}
			\label{eq8c}
			\begin{aligned}
				\gamma_{X}(h) &= Cov[X_t,X_{t-h}]\\
				&= Cov[\lambda_t,X_{t-h}]\\
				&= \sum_{i=1}^p\alpha_iCov[X_{t-i},X_{t-h}] + \sum_{j=1}^q \beta_j Cov[\lambda_{t-j},X_{t-h}]\\
				&= \sum_{i=1}^{p}\alpha_iCov[X_{t-i},X_{t-h}]+\sum_{j=1}^{min(h-1,q)}\beta_jCov[X_{t-j},X_{t-h}] + \sum_{j=h}^q\beta_jCov[\lambda_{t-j},\lambda_{t-h}].
			\end{aligned}
		\end{equation} 
	\end{proof}
	\subsection{Plots from the Simulation study}\label{plot}
	\subsubsection{Histograms of posterior distributions, traceplots and autocorrelation of parameters for configuration (I)}
	\begin{figure}[H]
		\centering
		\includegraphics[scale=0.4]{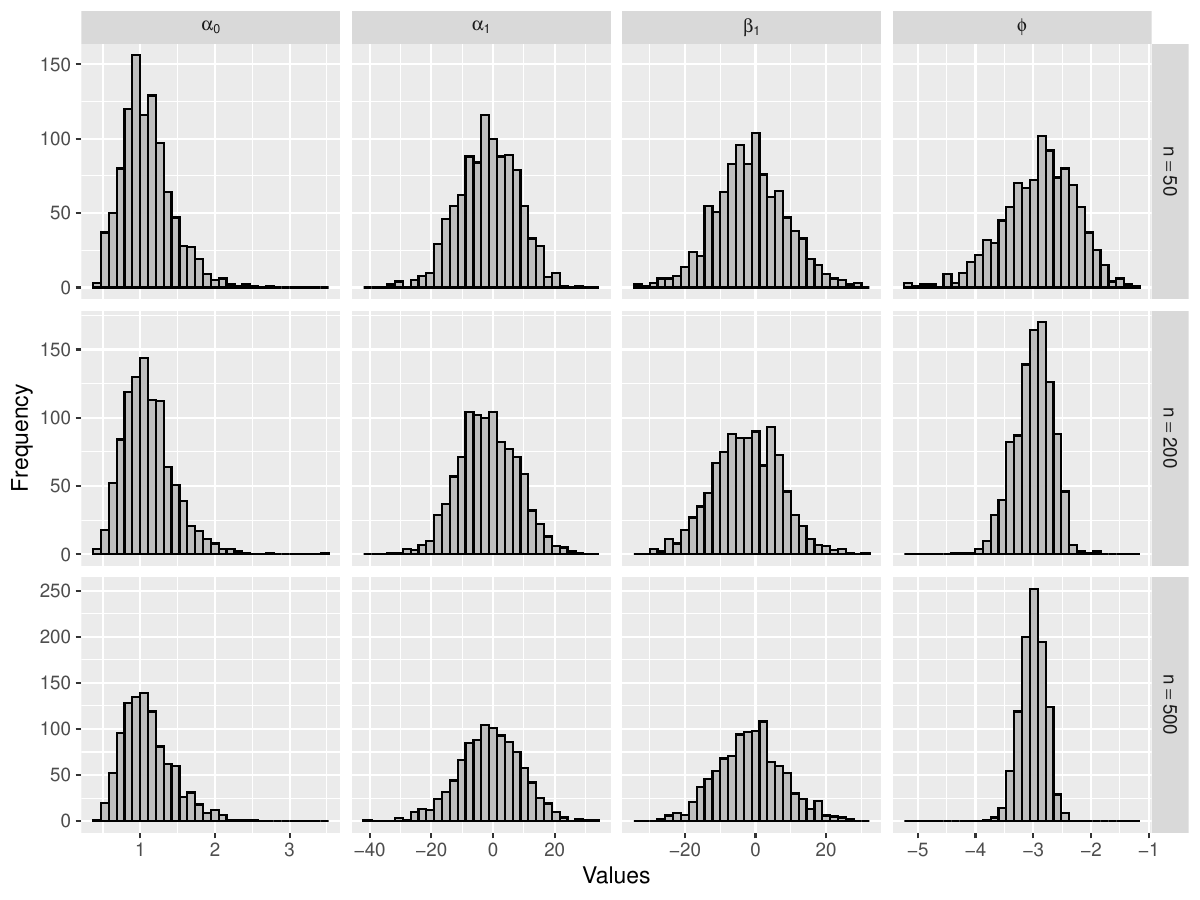}
		\caption{Histograms of posterior distributions for parameters, $\alpha_0$, $\alpha_1$, $\beta_1$ and $\phi$ for various sample sizes simulated from parameter configuration (I). }
		\label{fig:1}
	\end{figure}
	
	\begin{figure}[H]
		\centering
		\includegraphics[scale=0.4]{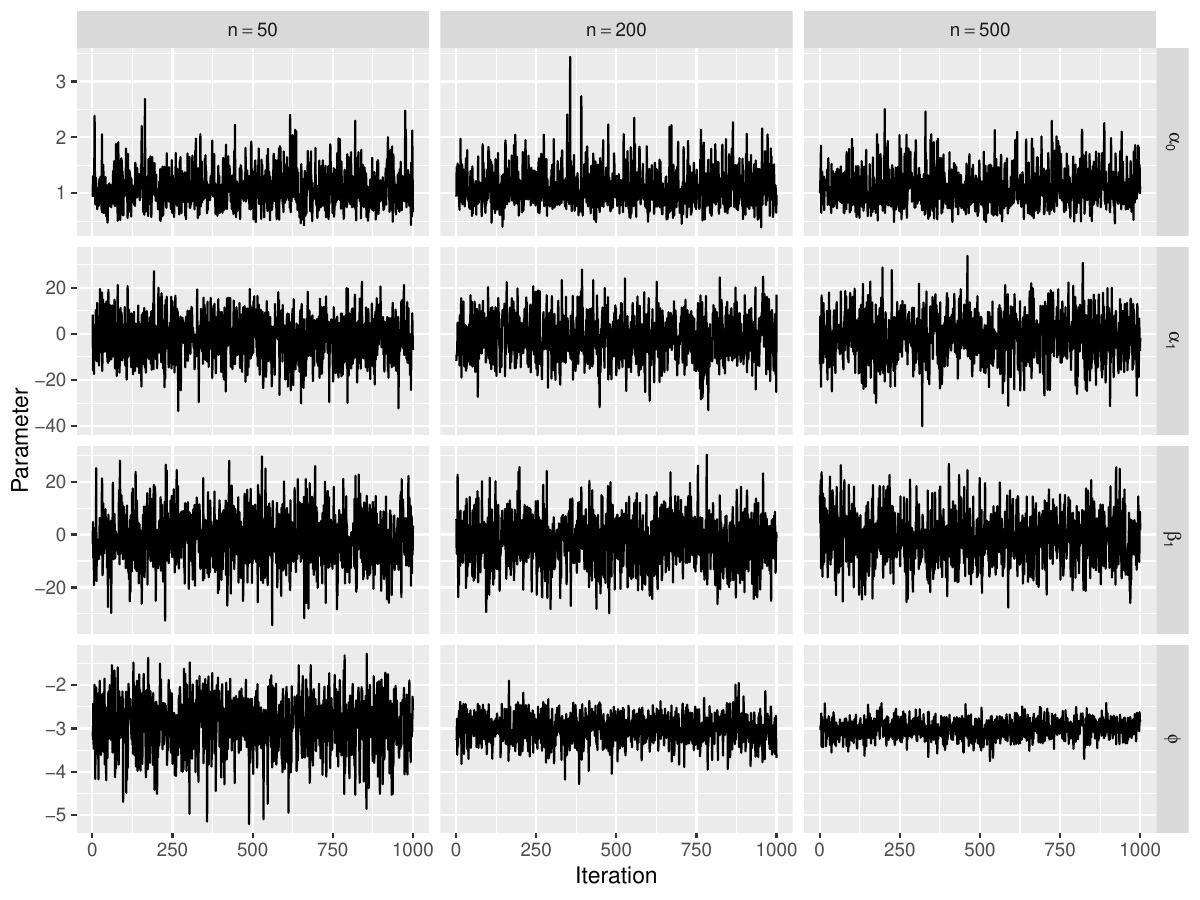}
		\caption{Traceplots for parameters, $\alpha_0$, $\alpha_1$, $\beta_1$ and $\phi$ for various sample sizes simulated from parameter configuration (I).}
		\label{fig:2}
	\end{figure}
	\begin{figure}[H]
		\centering
		\includegraphics[scale=0.5]{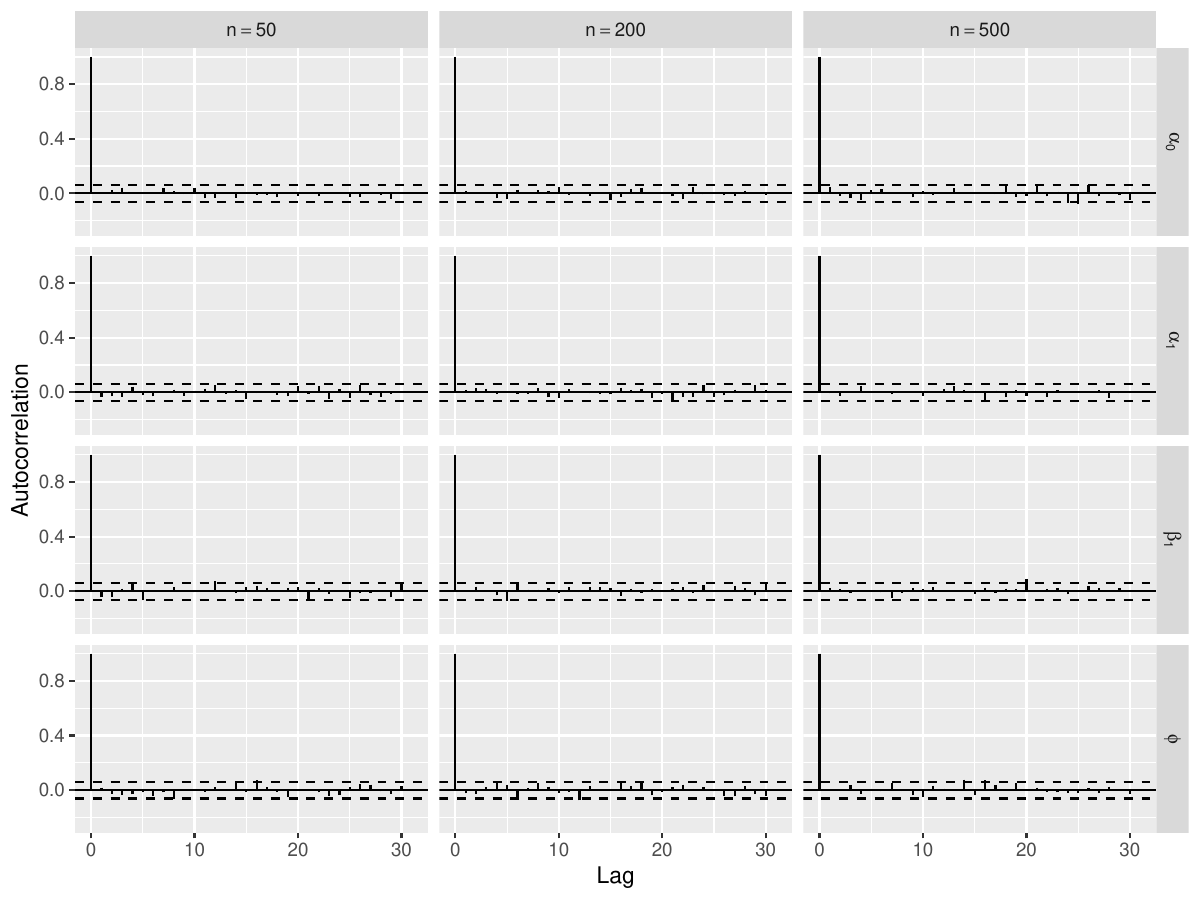}
		\caption{Autocorrelation for parameters, $\alpha_0$, $\alpha_1$ $\beta_1$ and $\phi$ for various sample sizes simulated from parameter configuration (I). }
		\label{fig:3}
	\end{figure}
	\footnotesize{NOTE: The plots are presented for the case of transformed unconstrained parameters. That is, since some of the parameters are constrained, as mentioned in \textcolor{blue}{Subsection} \ref{hmc}, we have applied logit transformation for them to vary along the real line. For example, in the case of Model (I) where $\phi = 0.05$ is bounded between 0 and 1, after applying the transformation, the unconstrained parameter has the value --2.94.}
\subsubsection{Histograms of posterior distributions, traceplots and autocorrelation of parameters for configuration (II)}
	\begin{figure}[H]
		\centering
		\includegraphics[scale=0.4]{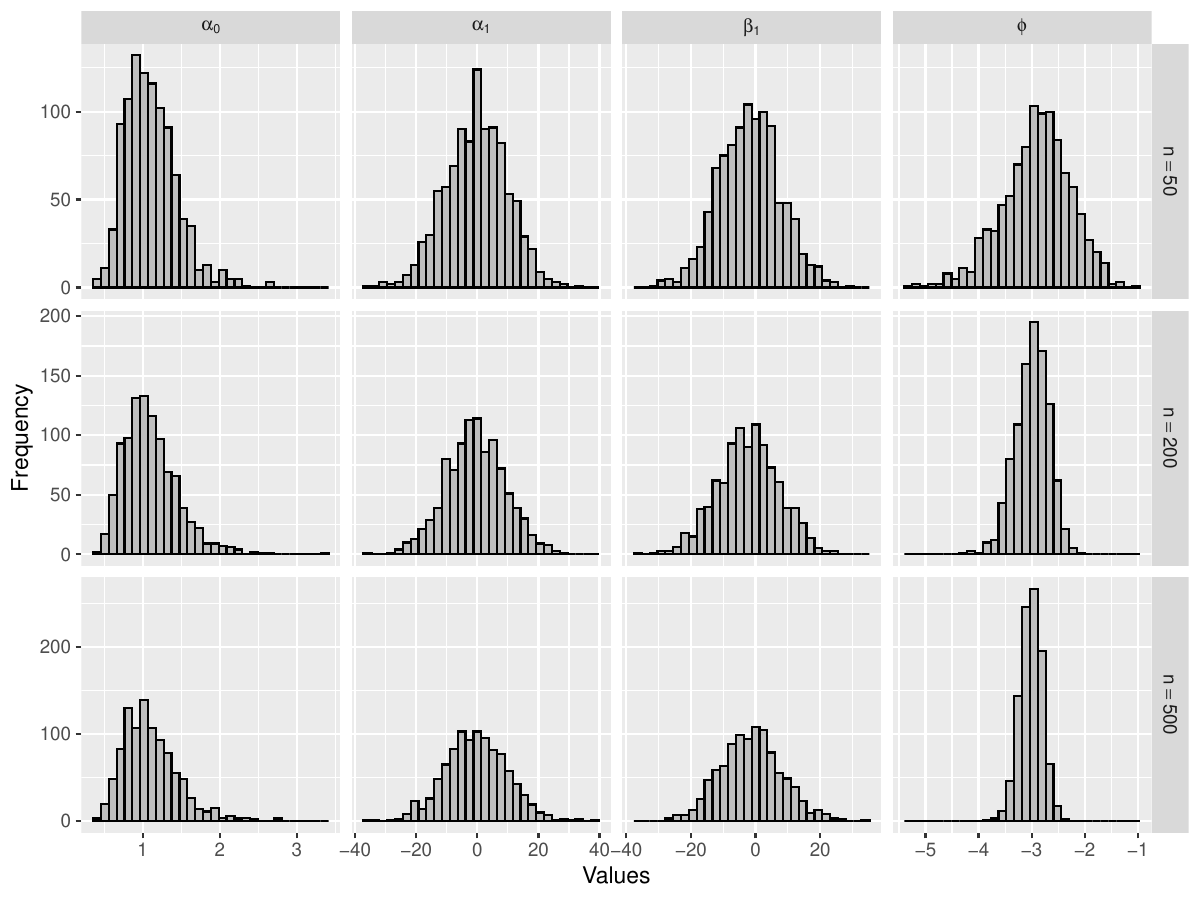}
		\caption{Histograms of posterior distributions for parameters, $\alpha_0$, $\alpha_1$, $\beta_1$ and $\phi$ for various sample sizes simulated from parameter configuration (II).}
		\label{fig:4}
	\end{figure}
	\begin{figure}[H]
		\centering
		\includegraphics[scale=0.5]{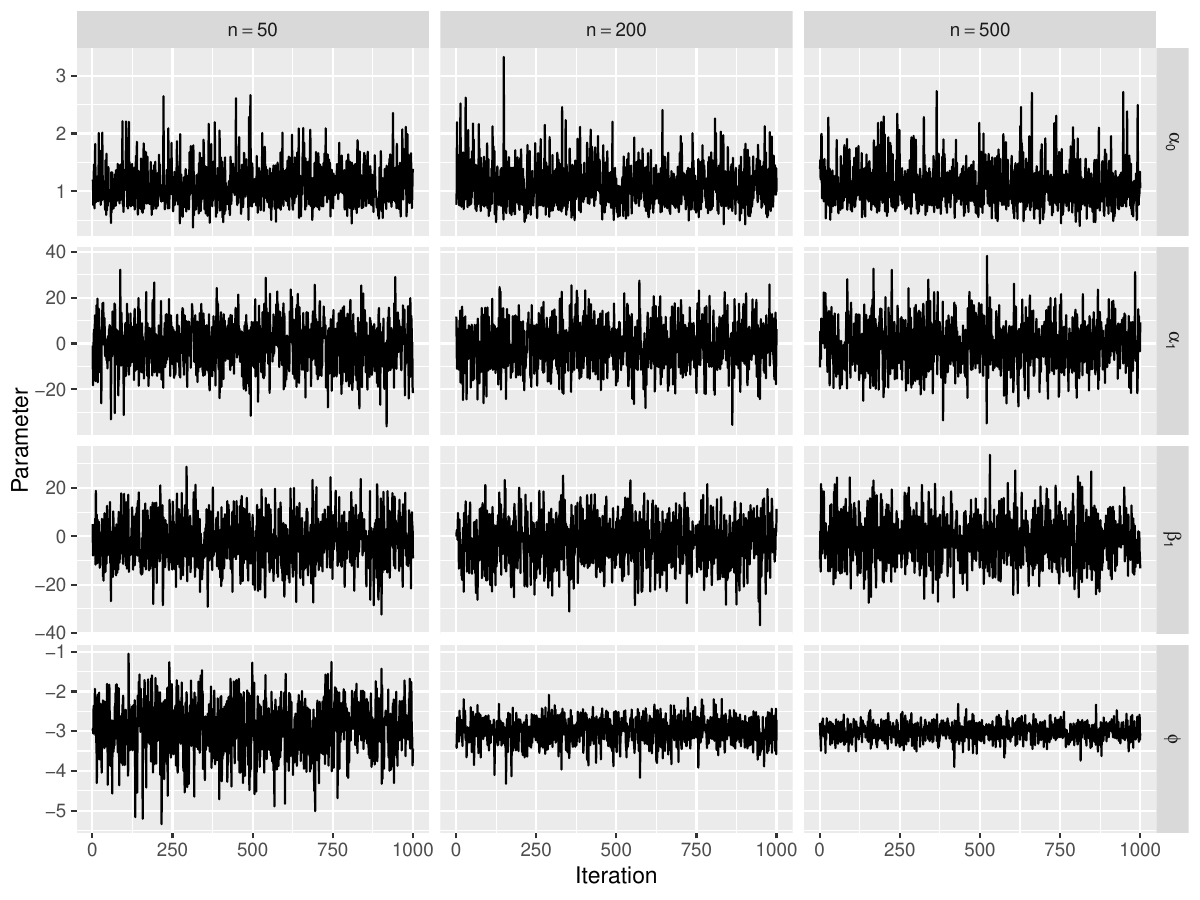}
		\caption{Traceplots for parameters, $\alpha_0$, $\alpha_1$, $\beta_1$ and $\phi$ for various sample sizes simulated from parameter configuration (II). }
		\label{fig:5}
	\end{figure}
	\begin{figure}[H]
		\centering
		\includegraphics[scale=0.5]{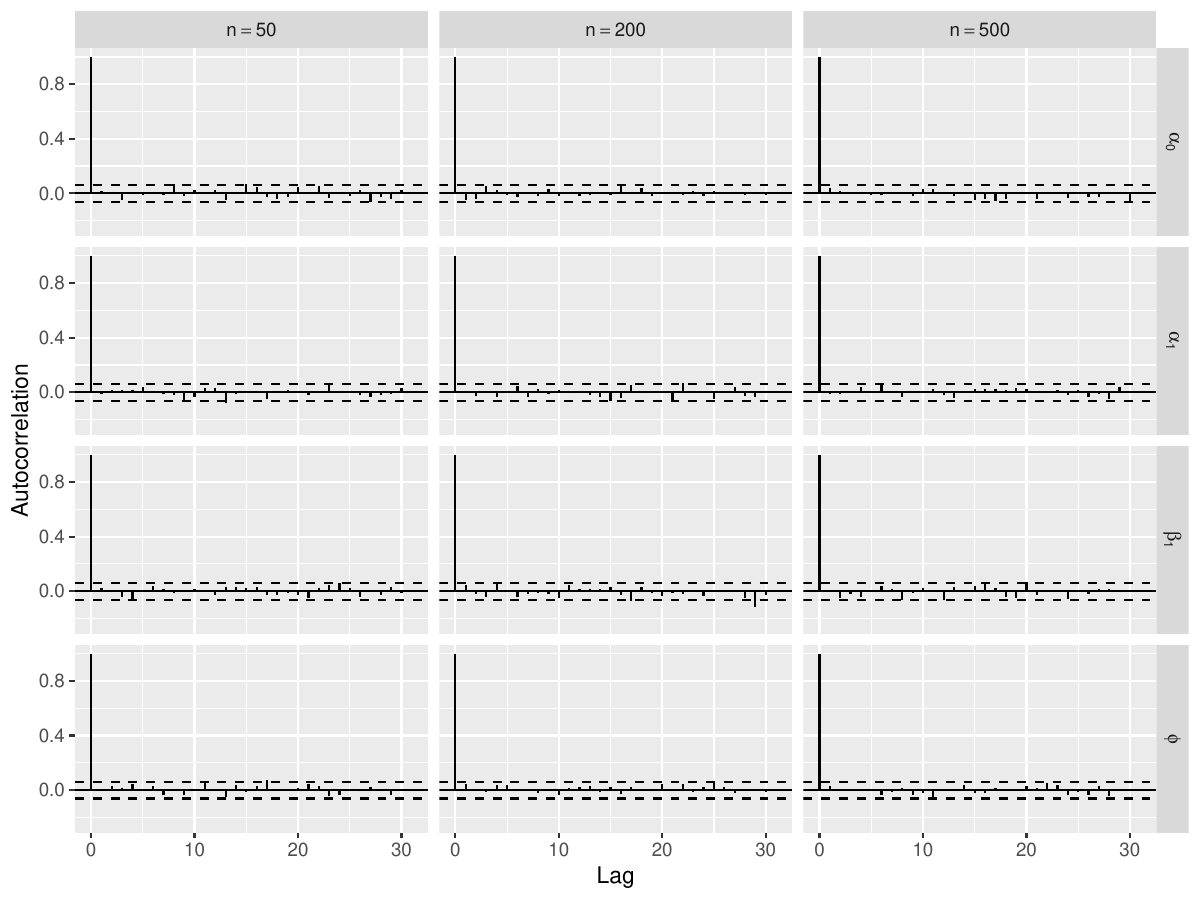}
		\caption{Autocorrelation for parameters, $\alpha_0$, $\alpha_1$, $\beta_1$ and $\phi$ for various sample sizes simulated from parameter configuration (II). }
		\label{fig:6}
	\end{figure}
	\subsubsection{Histograms of posterior distributions, traceplots and autocorrelation of parameters for configuration (III)}
	\begin{figure}[H]
		\centering
		\includegraphics[scale=0.5]{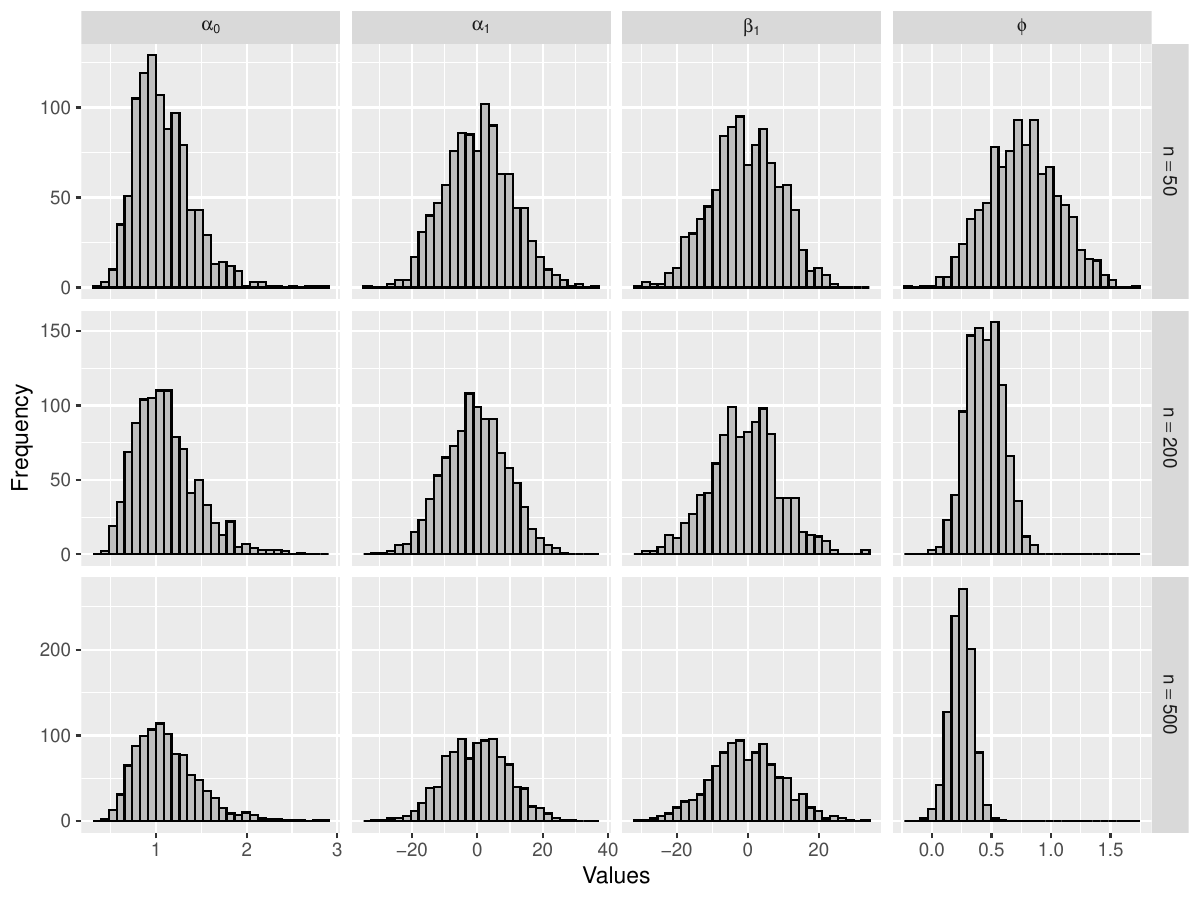}
		\caption{Histograms of posterior distributions for parameters, $\alpha_0$, $\alpha_1$, $\beta_1$ and $\phi$ for various sample sizes simulated from parameter configuration (III). }
		\label{fig:7}
	\end{figure}
	\begin{figure}[H]
		\centering
		\includegraphics[scale=0.5]{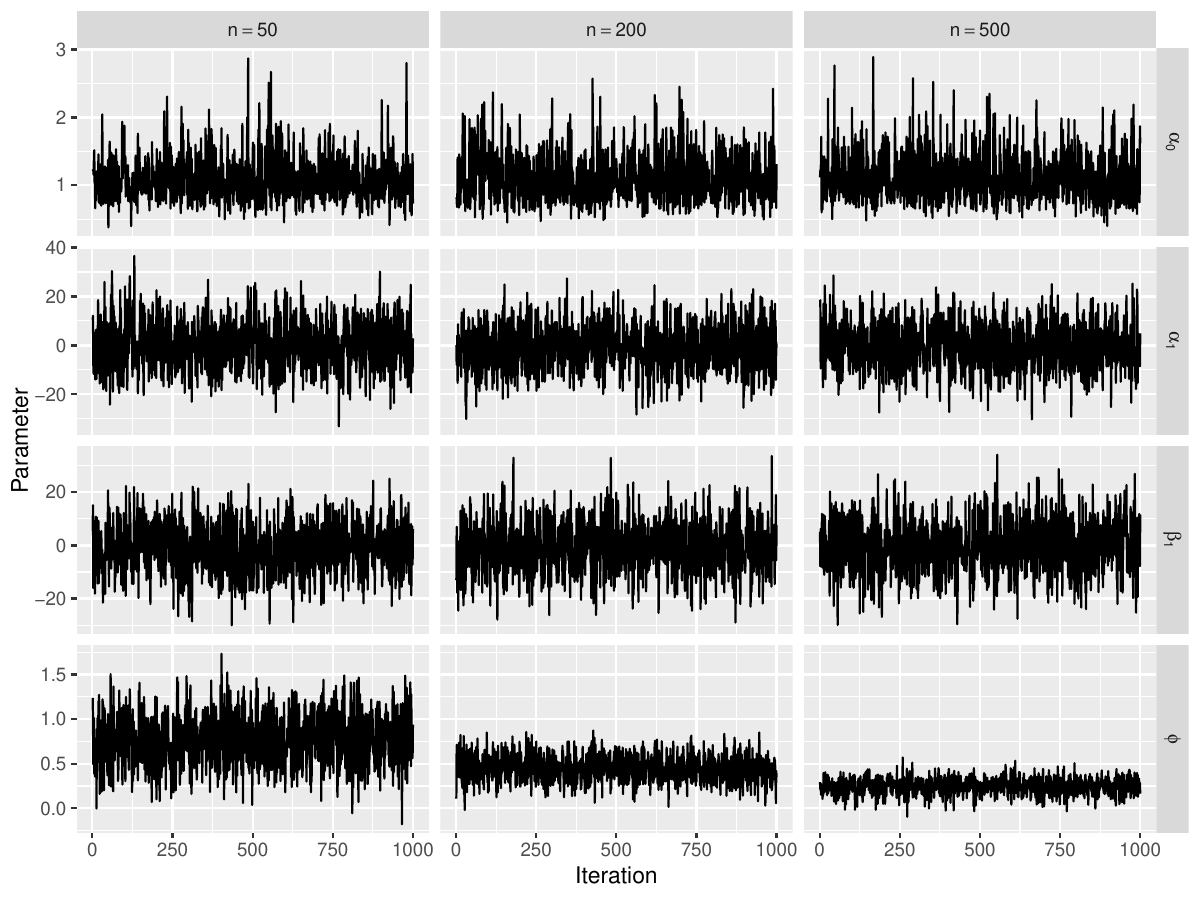}
		\caption{Traceplots for parameters, $\alpha_0$, $\alpha_1$, $\beta_1$ and $\phi$ for various sample sizes simulated from parameter configuration (III).}
		\label{fig:8}
	\end{figure}
	\begin{figure}[H]
		\centering
		\includegraphics[scale=0.5]{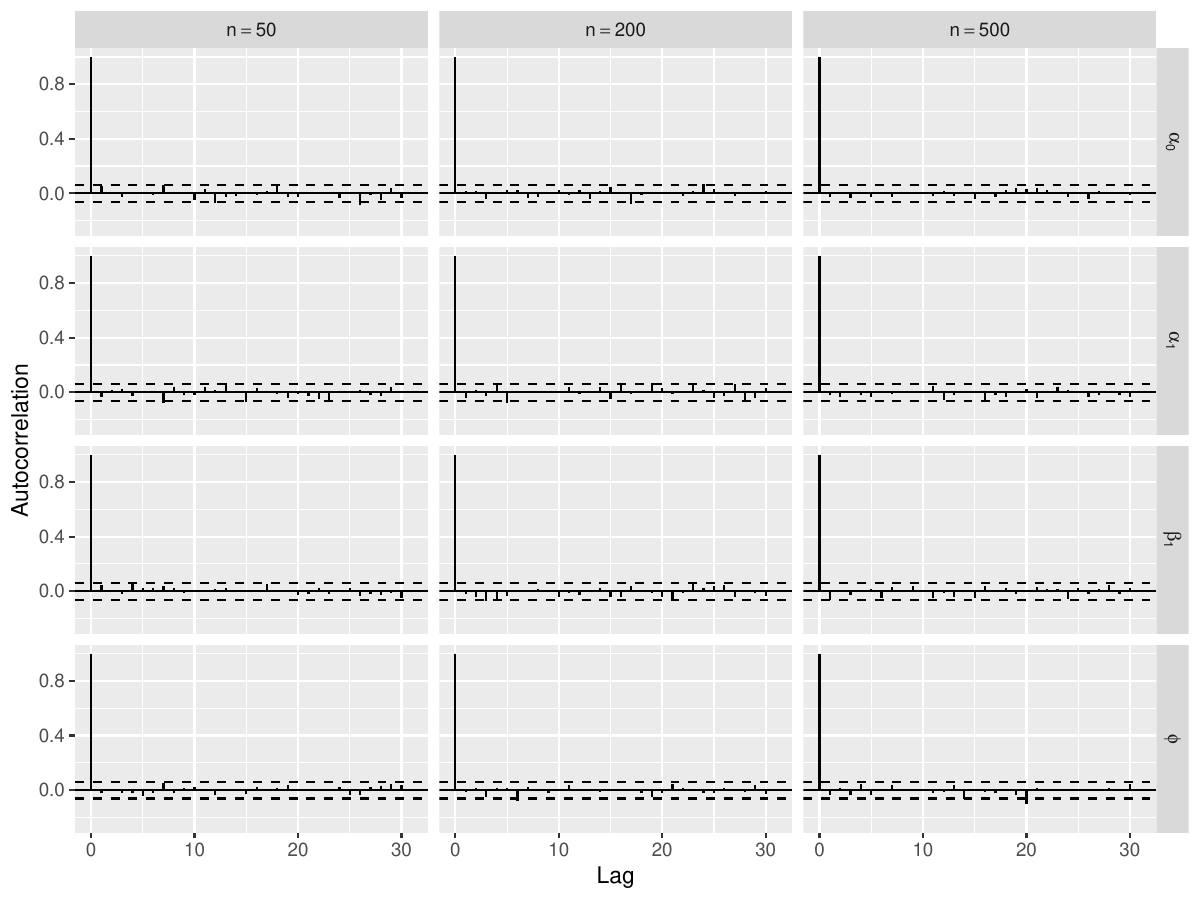}
		\caption{Autocorrelation for parameters, $\alpha_0$,  $\alpha_1$, $\beta_1$ and $\phi$ for various sample sizes simulated from parameter configuration (III). }
		\label{fig:9}
	\end{figure}		
	
		\subsubsection{Histograms of posterior distributions, traceplots and autocorrelation of parameters for configuration (IV)}
	\begin{figure}[H]
		\centering
		\includegraphics[scale=0.5]{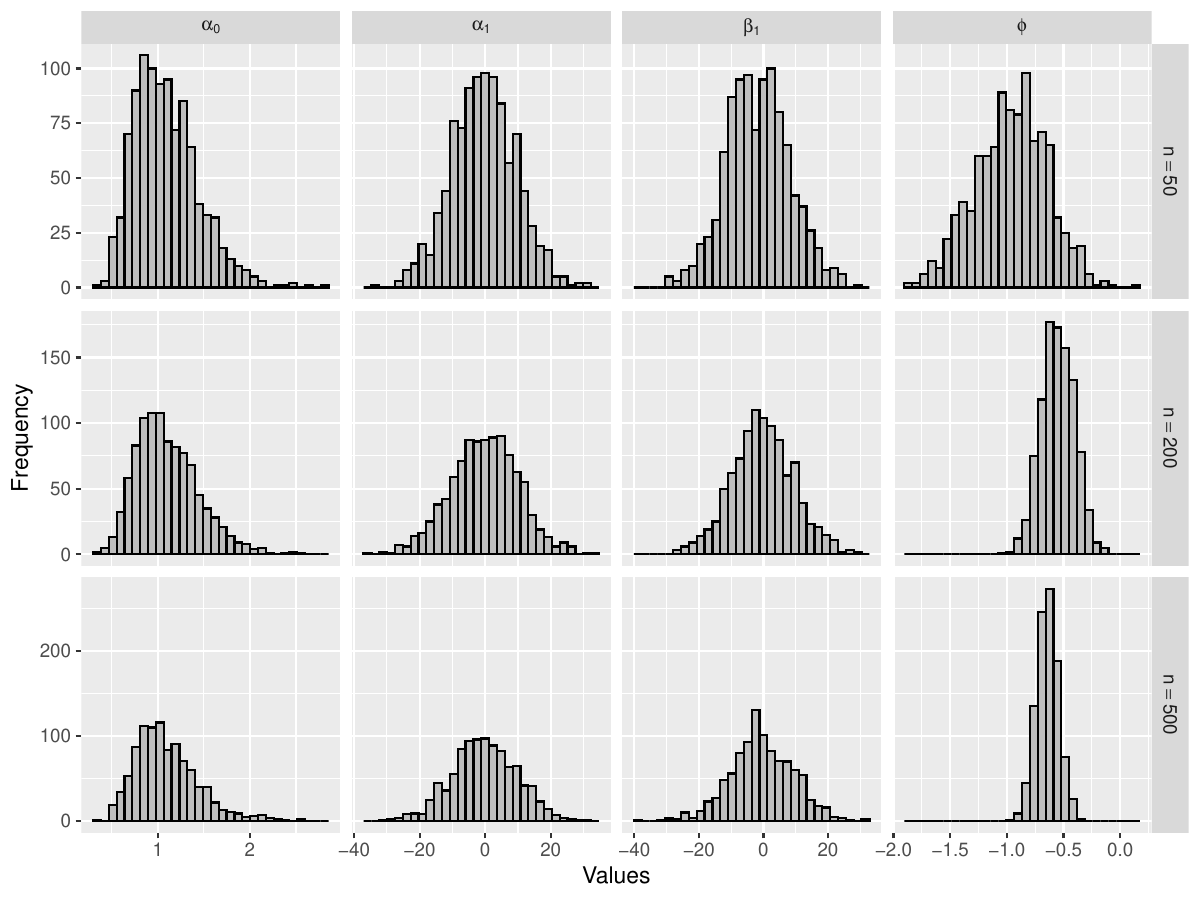}
		\caption{Histograms of posterior distributions for parameters, $\alpha_0$, $\alpha_1$, $\beta_1$ and $\phi$ for various sample sizes simulated from parameter configuration (IV). }
		\label{fig:10}
	\end{figure}
	\begin{figure}[H]
		\centering
		\includegraphics[scale=0.5]{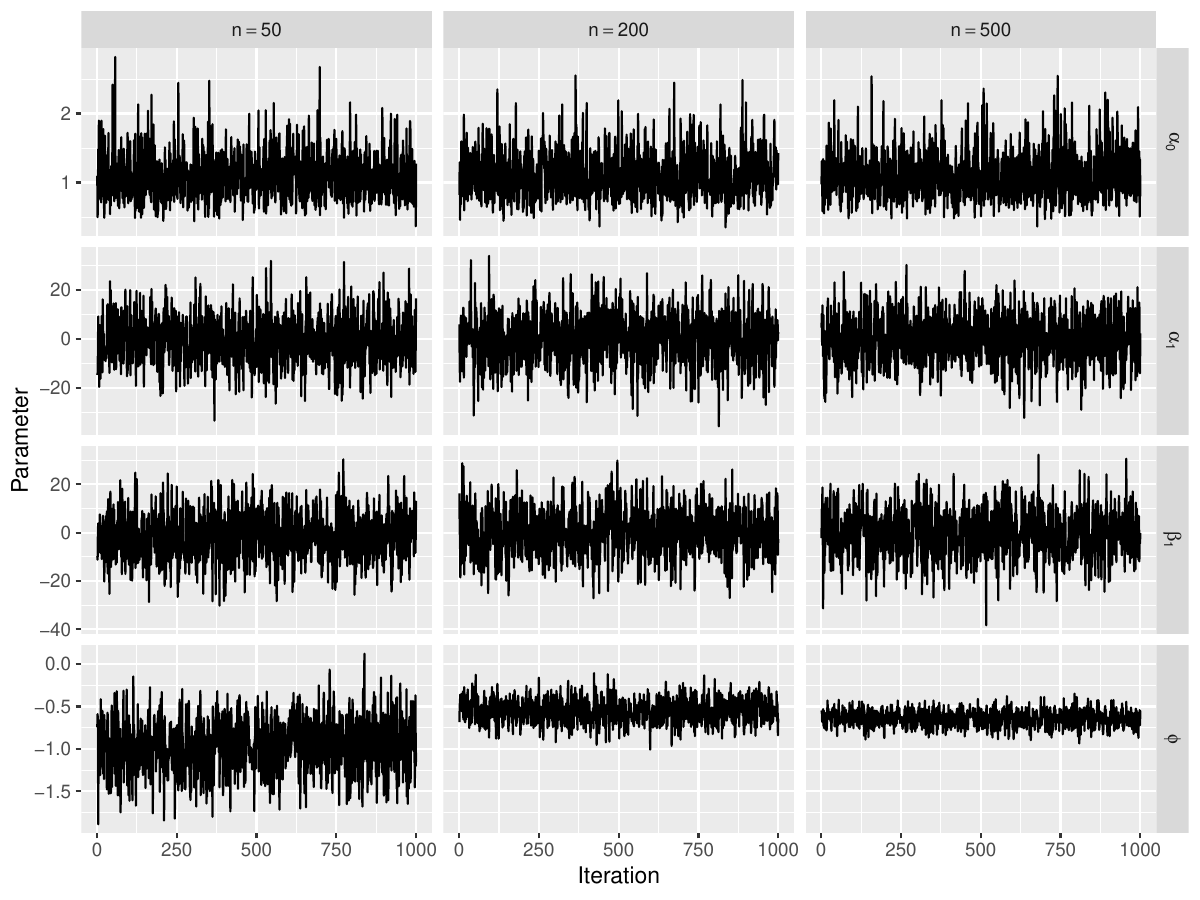}
		\caption{Traceplots for parameters, $\alpha_0$, $\alpha_1$, $\beta_1$ and $\phi$ for various sample sizes simulated from parameter configuration (IV).}
		\label{fig:11}
	\end{figure}
	\begin{figure}[H]
		\centering
		\includegraphics[scale=0.5]{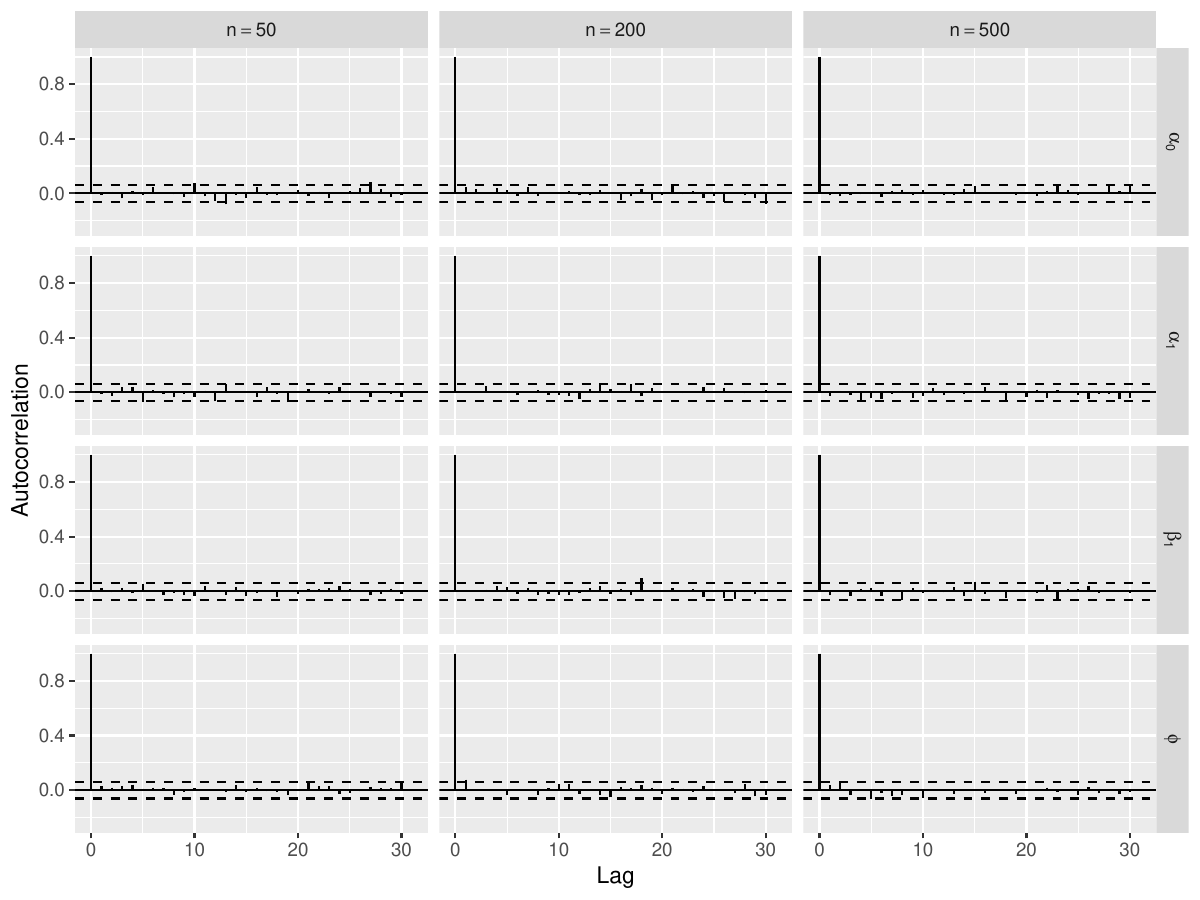}
		\caption{Autocorrelation for parameters, $\alpha_0$,  $\alpha_1$, $\beta_1$ and $\phi$ for various sample sizes simulated from parameter configuration (IV). }
		\label{fig:12}
	\end{figure}		
	
\end{appendices}

\end{document}